\documentclass[12pt]{extarticle}
\usepackage{amsmath}
\usepackage{amsthm, amssymb, hyperref, color}
\usepackage[shortlabels]{enumitem}
\usepackage{graphicx}
\usepackage[all]{xypic}
\usepackage{makecell}
\usepackage[final]{pdfpages}
\usepackage{amsfonts}
\setboolean{@twoside}{false}
\usepackage{pdfpages}
\usepackage{subfigure}
\usepackage{scalefnt}
\usepackage{verbatim}
\usepackage{color}
\oddsidemargin \evensidemargin
\date{}

\newtheorem{theorem}{Theorem}
\numberwithin{theorem}{section}
\newtheorem{proposition}[theorem]{Proposition}
\newtheorem{lemma}[theorem]{Lemma}
\newtheorem{corollary}[theorem]{Corollary}

\newtheorem{problem}[theorem]{Problem}

\newtheorem{example}[theorem]{Example}

\newcommand{\RR}{\mathbb{R}}

\definecolor{g4}{rgb}{0,0.4,0}

\newcommand{\diag}[1]{\mbox{diag}{#1}}

\title{\textbf{Geometry of  Log-Concave Density Estimation}}
\author{Elina Robeva, Bernd Sturmfels,  and Caroline Uhler}
\begin{document}

\maketitle

\begin{abstract}
\noindent 
Shape-constrained density estimation is an important topic in mathematical statistics. 
We focus on densities on $\mathbb{R}^d$ that are log-concave, and we
study geometric properties of the maximum likelihood estimator (MLE) for
weighted samples.
Cule, Samworth, and Stewart showed that the logarithm of the optimal 
log-concave density is piecewise linear and supported on a regular subdivision of the samples. This defines a map from the space of weights
to the set of regular subdivisions of the samples, i.e. the face poset of their secondary polytope. We prove that this map is surjective. In fact, every regular subdivision arises in the MLE for some set of weights with positive probability, but coarser subdivisions appear to be more likely to arise than finer ones. To quantify these results, we introduce a continuous version of the secondary polytope, whose dual we name the Samworth body. This article establishes a new link between
geometric combinatorics and nonparametric statistics,
and it suggests numerous open problems.
\end{abstract}

\section{Introduction}

Let $X = (x_1,x_2,\ldots,x_n)$ be a
configuration of $n$ distinct labeled points in $\RR^d$,
and let $w = (w_1,w_2,\ldots,w_n)$ be a vector
of positive weights that satisfy
$w_1 + w_2 + \cdots + w_n  =1$.
The pair $(X,w)$ is our dataset.
Think of experiments whose outcomes are measurements in $\RR^d$.
We interpret $w_i$ as the fraction among our experiments that led to
the sample point $x_i$ in $\RR^d$.

From this dataset one can compute
the sample mean  $\,\hat \mu = \sum_{i=1}^n w_i x_i \,$
and the sample covariance matrix
$\,\hat \Sigma = \sum_{i=1}^n w_i (x-\hat \mu) (x - \hat \mu)^T $.
Suppose that $\hat \Sigma$ has full rank $d$ and we wish to
approximate the sample distribution by a Gaussian with density $f_{\mu,\Sigma}$ on $\RR^d$.
Then~$(\hat \mu, \hat \Sigma) $ is the best solution
in the likelihood sense, i.e.~this pair
 maximizes the log-likelihood function
\begin{equation}
\label{ex:loglikelihood}  (\mu,\Sigma) \,\,\, \mapsto \,\,\, \sum_{i=1}^n w_i \cdot 
{\rm log}(f_{\mu,\Sigma} (x_i)) .
\end{equation}

 In  {\em nonparametric statistics} one abandons the assumption that the
desired probability density  belongs to a model with finitely many parameters.
Instead one seeks to maximize 
\begin{equation}
\label{eq:loglikelihood2}
f \,\,\,\mapsto \,\,\, \sum_{i=1}^n w_i \cdot {\rm log}(f (x_i))  
\end{equation}
over all density functions $f$. However,
since $f$ can be chosen arbitrarily close to the
finitely supported measure $\sum_{i=1}^n w_i \delta_{x_i}$, it is necessary to put constraints on $f$. 
One approach to~a meaningful maximum likelihood  problem is to impose
{\em shape constraints} on
the graph of~$f$. This line of research started with Grenander~\cite{Grenander}, who analyzed the 
case when the density is monotonically decreasing. Another popular shape constraint is convexity of the density~\cite{Groeneboom}.


 In this paper, we consider maximum likelihood estimation, under the assumption that $f$ is {\em log-concave}, i.e. that~${\rm log}(f)$ is a concave
function from $\RR^d$ to $\RR \cup \{-\infty\}$. Density estimation under log-concavity has been studied in depth in recent years; see e.g.~\cite{CSS, Duembgen, Walther}. Note that
Gaussian distributions $f_{\mu,\Sigma}$ are log-concave. Hence,
the following optimization
problem naturally generalizes the familiar task of maximizing (\ref{ex:loglikelihood})
over all pairs of parameters $(\mu,\Sigma)$:
\begin{equation}
\label{eq:ourproblem}
\begin{matrix}
\hbox{Maximize the log-likelihood (\ref{eq:loglikelihood2}) of the given sample $(X,w)$ over all} \\
\hbox{integrable functions $f: \RR^d \rightarrow \RR_{\geq 0} $ such that
${\rm log}(f)$ is concave and
$\,\int_{\RR^d} f(x) dx = 1$.}
\end{matrix}
\end{equation}

A solution to this optimization problem was given by
Cule, Samworth and Stewart  in \cite{CSS}.
They showed that the logarithm of the optimal density
$\hat f$ is a piecewise linear concave function,
whose regions of linearity are the cells of a
regular polyhedral subdivision of the configuration $X$.
This reduces the infinite-dimensional optimization
problem (\ref{eq:ourproblem}) to a convex optimization
problem in $n$ dimensions, since $\hat f$ is uniquely defined once its values at $x_1,\dots, x_n$ are known. An efficient algorithm for solving
this problem is described in \cite{CSS}. It is implemented
in the {\tt R} package {\tt LogConcDEAD} due to Cule, Gramacy and Samworth \cite{CGS}.

\begin{figure}[b!]
\begin{center}
\includegraphics[width = 0.39\textwidth]{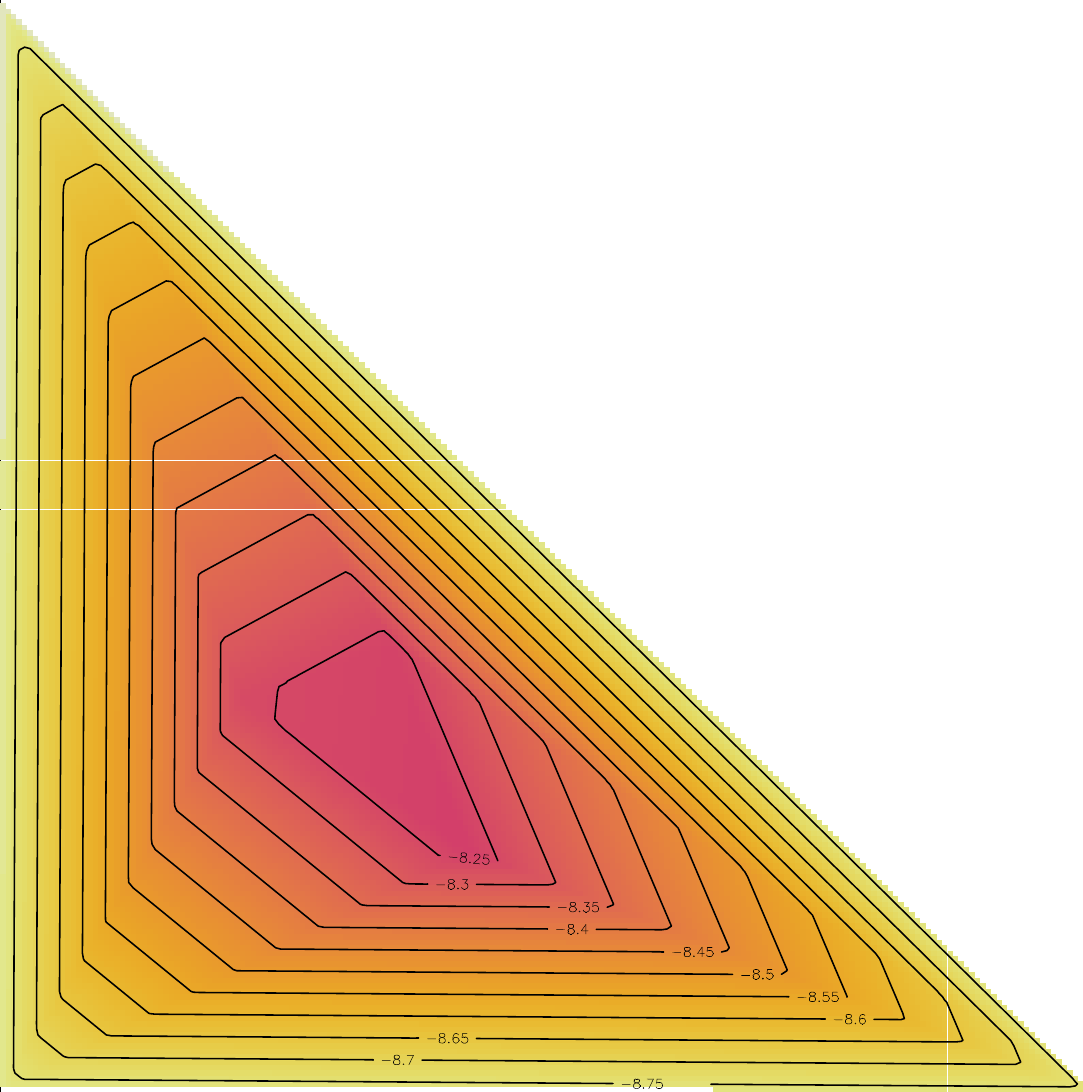} $\quad\quad$
\includegraphics[width = 0.39\textwidth]{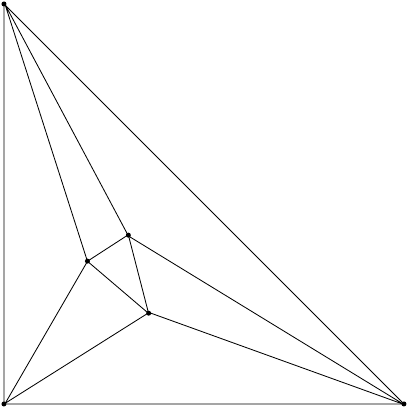}
\end{center}
\vspace{-0.2in}
\caption{The optimal log-concave density $\hat f$ for the  six data points in (\ref{eq:sixpoints}) with unit weights.
The graph of the piecewise linear concave function $\log(\hat f)$ is shown on the left.
The regions of linearity are the seven triangles in
  the triangulation of the six points shown on the right.
\label{fig:octahedron}}
\end{figure}

\begin{example} \label{ex:octahedron} \rm
Let $d=2$, $n=6$, $w = \frac{1}{6}(1,1,1,1,1,1)$, and fix the  point configuration
\begin{equation}
\label{eq:sixpoints}
 X \,\,=\,\, \bigl( \,(0, 0)\,, \,(100, 0)\,,\, (0, 100)\,,\, (22, 37)\,, \,(43, 22)\,, \,(36, 41) \,\bigr). 
 \end{equation}
 The graphical output generated by  {\tt LogConcDEAD} is shown on  the left in Figure~\ref{fig:octahedron}.
 This is the graph  of the function ${\rm log}(\hat f)$ that solves  (\ref{eq:ourproblem}). 
This piecewise linear concave function has seven linear pieces, namely the triangles 
on the right in Figure~\ref{fig:octahedron}, with vertices taken from~$X$.
 \end{example}

The purpose of this paper is to establish a link
between nonparametric statistics and geometric combinatorics.
We develop a generalization of the theory of
regular triangulations
arising in the context of maximum likelihood estimation
for log-concave densities.

A key feature of our approach is that we emphasize the  presence 
of unequal weights in the maximum likelihood estimation problem
(\ref{eq:ourproblem}). In fact, it is important for us that the weights
$w$ are allowed to vary.  While this is a natural assumption from 
the perspective of geometry and combinatorics, it is
also well-motived by statistics. For instance, unequal weights
are necessary when applying the bootstrap to assess uncertainty in the density estimate.
In addition, Leister \cite[Section 2.3]{Lei} studies 
Hidden Markov Models for state-dependent log-concave densities.
Her EM algorithm  solves the problem (\ref{eq:ourproblem}) repeatedly
for different weights.

Our paper is organized as follows. In Section 2 we first
review the relevant mathematical concepts, especially
polyhedral subdivisions and secondary polytopes \cite{DRS, GKZ}.
We then generalize results in \cite{CSS} from the case
of unit weights $w = \frac{1}{n}(1,1,\ldots,1)$ to arbitrary weights~$w$. Theorem~\ref{thm:samworth} 
casts the problem (\ref{eq:ourproblem})  as a linear optimization problem
 over a convex subset $\mathcal{S}(X)$ of $\RR^n$,
which we call the {\em Samworth body} of~$X$.
Theorem \ref{thm:IntegralFormula} 
uses integrals as in \cite{Bar} to give
an unconstrained formulation of this problem
with an explicit objective function.

Cule, Samworth and Stewart \cite{CSS} discovered that log-concave density estimation 
leads to regular polyhedral subdivisions. In this paper we prove the following
converse to their result:

\begin{theorem}
\label{thm:converse}
Let $\Delta$ be any regular polyhedral subdivision of the configuration $X$.
There exists a non-empty open subset $\,\mathcal{U}_\Delta$ in $\RR^n$
such that, for every $w \in \mathcal{U}_\Delta$, the optimal solution 
$\hat f$ to (\ref{eq:ourproblem}) is a piecewise log-linear function whose
regions of linearity are the cells of~$\Delta$.
\end{theorem}

The proof of Theorem~\ref{thm:converse} appears in Section 3.
We introduce a remarkable symmetric function $H$ that serves as a key technical tool.
The theory behind $H$ seems interesting in its own right.
In Theorem~\ref{thm:normalcone}
we characterize the normal cone at any boundary point of the Samworth body.
In other words, for a given concave piecewise log-linear function $f$,
we determine the set of all weight vectors $w$ such that
$f$ is the optimal solution in (\ref{eq:ourproblem}).

In Section 4  we view (\ref{eq:ourproblem}) as a parametric optimization problem,
as either $w$ or $X$ vary. Variation of $w$ is explained by
the geometry of the Samworth body. We
explore empirically the probability that a given subdivision is
optimal. We observe that triangulations are rare.
Thus pictures like the triangulation in Figure \ref{fig:octahedron} are
exceptional and deserve special attention.

In Section 5 we focus our attention on the case of unit weights,
and we examine the constraints this imposes on $\Delta$.
Theorem \ref{thm:d+2points} shows that triangulations
never occur for $n=d+2$ points in $\RR^d$ with unit weights.
A converse to this result is established in Theorem~\ref{thm:d+3points}.

Sections 4 and 5 conclude with several open problems.
These suggest possible lines of inquiry for a future research theme
that  might be named {\em Nonparametric Algebraic Statistics}.

\section{Geometric Combinatorics}

We begin by reviewing concepts from geometric combinatorics,
studied in detail in the books by
De Loera, Rambau and Santos \cite{DRS}
and Gel'fand, Kapranov and Zelevinsky~\cite{GKZ}.
See Thomas \cite[\S 7-8]{Tho} for a first introduction.
Let $X = (x_1,\ldots,x_n)$ be a configuration as before
and $P = {\rm conv}(X)$ its convex hull in $\RR^d$.
We assume that the polytope $P$ has dimension $d$.

Fix a real vector $y = (y_1,\ldots,y_n)$. We write
$h_{X,y}$ for the smallest concave function $h$ on $\RR^n$
such that $h(x_i) \geq y_i$ for $i=1,\ldots,n$. The graph of $h_{X,y}$
is the upper convex hull of $\{(x_1,y_1),\ldots,(x_n,y_n)\}$ in $\RR^{n+1}$.
Hence $h_{X,y}(t)$ is the largest real number $h^*$ such 
that $(t,h^*)$ is in the convex hull of $\{(x_1,y_1),\ldots,(x_n,y_n)\}$.
In particular, $h_{X,y}(t) = - \infty$ for $t \not\in P$.
Up to sign, the function $h_{X,y}$ is called
the {\em characteristic section}
in \cite[Definition 5.2.12]{DRS}.
We also refer to $h_{X,y}$ as the {\em tent function},
with (some of) the points $(x_i,y_i)$ being the {\em tent poles}.
The vector $y$ is called {\em relevant} if $h_{X,y}(x_i) = y_i$
for $i=1,\ldots,n$, i.e.~if each $(x_i,y_i)$ is a tent pole.
This fails, for example, if $x_i$ lies in the interior of $P$ and $y_i$ is small
relative to the other~$y_j$.

A {\em regular subdivision} $\Delta$ of $X$ is a collection of subsets of $X$
whose convex hulls are the regions of linearity of the function $h_{X, y}$ for some $y\in \mathbb R^n$.
These regions are $d$-dimensional polytopes, and are called the {\em cells} of $\Delta$.
 A regular subdivision $\Delta$ is a {\em regular triangulation} of $X$ if each 
 cell is a $d$-dimensional simplex.
The {\em secondary polytope} $\Sigma(X)$ is a polytope
of dimension $n-d-1$ in $\RR^n$ whose faces are in bijection with
the regular subdivisions of $X$.
In particular, the vertices of $\Sigma(X)$ correspond to
the  regular triangulations of $X$;
see~\cite[\S 5]{DRS}.

If $\Delta$ is a regular triangulation of $X$,
then the $k$-th coordinate of the  corresponding vertex 
$z^\Delta$ of $\Sigma(X) \subset \RR^n$
is the sum of the volumes of all simplices in $\Delta$ that contain $x_k$.
In symbols,
\begin{equation}
\label{eq:GKZvector}   
z^\Delta_k \,\,= \,\, \sum_{\sigma \in \Delta:  \atop x_k \in \sigma} {\rm vol}(\sigma) .
\end{equation}
We call $z^\Delta = (z^\Delta_1,\ldots,z^\Delta_n) $ 
the {\em GKZ vector} of the triangulation $\Delta$, in reference to~\cite{GKZ}.

The support function of the secondary polytope $\Sigma(X)$ is the
piecewise linear function
$$ \RR^n \rightarrow \RR, \quad y \,\mapsto \, \int_P h_{X,y}(t) dt. $$
This follows from the equation in \cite[page 232]{DRS}.
The function is linear on each cone in the {\em secondary fan} of $X$.
For every $y$ in the secondary cone of a given regular triangulation $\Delta$, 
\begin{equation}
\label{eq:intislinear}
 \int_P h_{X,y}(t) dt \,\,= \,\, z^\Delta \cdot y \,\,=\,\, \sum_{i=1}^n z^\Delta_i y_i . 
 \end{equation}
This means that the convex dual to the secondary polytope has the representation
$$ \Sigma(X)^* \quad = \quad \bigl\{
y \in \RR^n \,:\, z^\Delta \cdot y \leq 1 \,\,\hbox{for all} \,\,\Delta \bigr\} \,\,\, = \,\,\,
\bigl\{ y \in \RR^n \,:\,\int_P h_{X,y}(t) dt \leq 1 \,\bigr\}. $$
Note that $\Sigma(X)^*$ is an unbounded polyhedron in $\RR^n$
since $\Sigma(X)$ has dimension $n-d-1$. Indeed,  $\Sigma(X)^*$ is the product of an
 $(n-d-1)$-dimensional polytope  and an orthant $\RR_{\geq 0}^{d+1}$.

\smallskip

We now introduce an object that looks like a continuous analogue
of $\Sigma(X)^*$. We define
\begin{equation}
\label{eq:samdef}
 \mathcal{S}(X) \quad = \quad \bigl\{\,
y \in \RR^n \,:\,\int_P {\rm exp}(h_{X,y}(t)) dt \leq 1 \,\bigr\}. 
\end{equation}

Inspired by \cite{CGS, CSS}, we 
call $\mathcal{S}(X)$ the {\em Samworth body} of the point configuration $X$.

\begin{proposition}
The Samworth body $\mathcal{S}(X)$ is a full-dimensional closed
 convex set in $\RR^n$. 
\end{proposition}

\begin{proof}
Let $y, y' \in \mathcal{S}(X)$ and consider
a convex combination $y'' = \alpha y + (1-\alpha)y'$ where $0 \leq \alpha \leq 1$.
For all $t \in P$, we have
$h_{X,y''}(t) \leq \alpha h_{X,y}(t) + (1-\alpha) h_{X,y'}(t)$, and therefore
$$ {\rm exp}(h_{X,y''}(t)) \,\leq \, {\rm exp} \bigl(\alpha h_{X,y}(t) + (1-\alpha) h_{X,y'}(t) \bigr) \,\leq \,
\alpha \cdot {\rm exp} ( h_{X,y}(t)) + (1-\alpha) \cdot {\rm exp}(h_{X,y'}(t) ) . $$
Now integrate both sides of this inequality over all $t \in P$.
The right hand side is bounded above by $1$, and hence so is the left hand side.
This means that $y'' \in \mathcal{S}(X)$. We conclude that $\mathcal{S}(X)$ is convex.
It is closed because the defining function is continuous, and it is $n$-dimensional
because all points $y$ whose $n$ coordinates are sufficiently negative lie in $\mathcal{S}(X)$.
\end{proof}

Every boundary point $y$ of the Samworth body $\mathcal{S}(X)$
defines a log-concave probability density function $f_{X,y}$ on $\RR^d$ that is supported
on the polytope $P = {\rm conv}(X)$. This density~is 
\begin{equation}
\label{eq:fydensity}
f_{X,y} \,\,: \,\,\,
 t \,\,\mapsto\,\, \begin{cases} {\rm exp}(h_{X,y}(t))  & {\rm if} \,\,\, t \in P, \\
\qquad 0  & {\rm otherwise}. \end{cases} 
\end{equation}

We fix a positive real vector 
$w = (w_1,\ldots,w_n) \in \RR^n_{\geq 0}$ that satisfies
$\sum_{i=1}^n w_i =1$. The following result rephrases the key
 results of Cule, Samworth and Stewart \cite[Theorems 2 and 3]{CSS},
 who proved this, in a different language, for the
 unit weight case $w = \frac{1}{n} (1,1,\ldots,1)$.

\begin{theorem}
\label{thm:samworth}
The linear functional $ \,y \mapsto w \cdot y = \sum_{i=1}^n w_i y_i$
is bounded above on the Samworth body $\,\mathcal{S}(X)$.
Its maximum over $\,\mathcal{S}(X)\,$ is attained at a unique point~$y^*$.
The corresponding log-concave density $f_{X,y^*}$ is the unique optimal solution to 
the estimation problem~(\ref{eq:ourproblem}).
\end{theorem}

\begin{proof}
We are claiming that  $\mathcal{S}(X)$ is strictly
convex and its recession cone is contained
in the negative orthant $\RR^n_{\leq 0}$.
The point $y^*$ represents the solution to
the optimization problem
\begin{equation}
\label{eq:constrained}
 \hbox{Maximize $\,w \cdot y\,$ subject to $\,y \in \mathcal{S}(X)$.} 
\end{equation}

The equivalence of (\ref{eq:ourproblem}) and (\ref{eq:constrained}) stems from
the fact that the optimal solution
$\hat f$ to the maximum likelihood problem  (\ref{eq:ourproblem}) has the form
$ f = f_{X, y}$ for some choice of $ y \in \RR^n$.
This was proven in \cite{CSS} for unit weights
$w = \frac{1}{n}(1,1,\ldots,1)$. The general case for positive weights $w\in\mathbb R^n$
follows analogously. To be precise, let $f$ be the maximum likelihood estimator for $(X,w)$, and
 let $y = \log(f(X))$ be the heights at the samples $X$. Then, the function 
$C \cdot f_{X,y}$, where $C\geq 1$ is the appropriate constant that makes
the integral equal to $1$,  gives a likelihood that is at least as large as $f$ does, 
with equality if and only if $f = f_{X,y}$. Thus, $f = f_{X,y}$.

Let $N$ be the sample size, so that $N_i = Nw_i$ is a positive integer for $i=1,\ldots,n$.
We think of $x_i$ as a sample point in $\RR^d$ that has been observed
$N_i$ times. If $f$ is any probability density function on $\RR^d$,
then the log-likelihood of the $N$ observations with respect to $f$ equals
\begin{equation}
\label{eq:MLF}
 N \cdot \sum_{i=1}^n w_i \cdot {\rm log}( f(x_i)).
 \end{equation}
 Maximizing (\ref{eq:MLF}) over log-concave densities is equivalent
 to maximizing  (\ref{eq:loglikelihood2}).
 We know from \cite[Theorem 2]{CSS} that the maximum is unique and is attained by
 $ f = f_{X,y^*}$ for some $y^* \in \RR^n$.
 Here $y^*$ is the unique relevant point in
$\,   \mathcal{S}(X) = \bigl\{ y \in \RR^n : \int_{\RR^d} f_{X,y}(t)dt \leq 1 \bigr\}\,$
that maximizes the linear functional $w \cdot y$.
Hence  (\ref{eq:ourproblem}) and (\ref{eq:constrained}) are equivalent for all $w \in \RR^n_{\geq 0}$.
 \end{proof}
 
The constrained optimization problem (\ref{eq:constrained})
can be reformulated as an unconstrained optimization problem.
For the unit weight case $w_1 = \cdots = w_n = 1/n$, this was done in
\cite[\S 3.1]{CSS}. This result can easily be extended to general weights.
In the language of convex analysis, Proposition \ref{prop:unconstrained}
says that the optimal value function of the convex optimization problem (\ref{eq:constrained}) is the
{\em Legendre-Fenchel transform} of the convex function $y \mapsto \int_P {\rm exp}(h_{X,y}(t)) dt$.

\begin{proposition}
\label{prop:unconstrained}
The constrained optimization problem {\rm (\ref{eq:constrained})} is equivalent to the unconstrained 
optimization problem
\begin{equation}
\label{eq:unconstrained}
 \hbox{{\rm Maximize} $\,\,w \cdot y - \int_{P} \exp(h_{X,y}(t)) dt\,\,$ over  all $\,\,y \in \RR^n$}, 
\end{equation}
where, as before, $P$ denotes the convex hull of $x_1, \dots , x_n\in\mathbb{R}^d$ and $h_{X,y}$ 
is the tent function, i.e.,~
$h_{X,y} : \mathbb{R}^d \to \mathbb{R}$ is the least concave function satisfying $h_{X,y}(x_i) \geq y_i$ for all
$i = 1,\ldots , n$. 
\end{proposition}

\begin{proof}
A proof for uniform weights  is given in \cite{CSS}. We here present
the proof for arbitrary weights $w_1,\ldots,w_n$. These are positive real numbers
that sum to $1$. This ensures that the objective function in (\ref{eq:constrained})
is bounded above, since the exponential term dominates when the coordinates
of $y$ become large. Clearly, the
optimum of (\ref{eq:constrained}) is attained on the boundary  $\partial \mathcal{S}(X)$
of the feasible set $\mathcal{S}(X)$, and we could equivalently optimize over that boundary.

 Now suppose that $y^*$ is an optimal solution of (\ref{eq:unconstrained}). 
 This implies that  $h_{X,y^*}(x_i) = y^*_i$, i.e.~each tent pole touches the tent.
 Otherwise $\,w \cdot y$ in the objective function can be increased 
 without changing $\int_{P} \exp(h_{X,y}(t)) dt$. Let $c:= \int_{P} \exp(h_{X,y^*}(t)) dt$. 
  We claim that $c=1$.
  
    Let $\hat{y}$ be a vector
    in $\RR^n$,   also satisfying $h_{X,\hat{y}}(x_i) = \hat{y}_i$ for all $i$,
     such that $\exp(h_{X,y^*}(t)) = c \exp(h_{X,\hat{y}}(t))$ and 
     $\int_{P} \exp(h_{X,\hat{y}}(t)) dt = 1$.
  This means that  $h_{X,y^*}(t) = \log(c) + h_{X,\hat{y}}(t)$ for all points $t$ in the polytope $ P$.
  In particular,  we have $ y_i^* -  \hat{y}_i = \log(c) $ for $i=1,2,\ldots,n$.
   
  We now analyze the difference of the objective functions at the points $\hat{y}$ and $y^*$:
$$w \cdot \hat{y} - \int_{P} \exp(h_{X,\hat{y}}(t)) dt - \left(w \cdot y^*
 - \int_{P} \exp(h_{X,y^*}(t)) dt\right) = -\log(c) -1+c.$$

Note that the function $c \mapsto -\log(c) -1+c$ is nonnegative.
Since $y^*$ maximizes $w \cdot y - \int_{P} \exp(h_{X,y}(t)) dt$, it follows that $-\log(c) -1+c = 0$, which implies  that $c=1$. So, the claim holds.
We have shown that the solution $y^*$ of (\ref{eq:unconstrained})
      also solves the following problem:
\begin{equation}
\label{eq:constrained_2}
 \hbox{Maximize $\,w \cdot y - \int_{P} \exp(h_{X,y}(t)) dt$ \; subject to\;  $\int_{P} \exp(h_{X,y}(t)) dt = 1$.} 
\end{equation}
But this  is equivalent to the constrained formulation (\ref{eq:constrained}), and the proof is complete.
\end{proof}

The objective function in (\ref{eq:unconstrained}) looks complicated because of the 
integral and because $h_{X,y}(t)$ depends piecewise linearly  on  both $y$ and $t$.
To solve our optimization problem,  a more explicit form is needed.
This was derived by Cule, Samworth and Stewart in \cite[Section B.1]{CSS}.
The formula that follows writes the objective function 
locally as an exponential-rational function.  This can also be derived
from work on polyhedral residues due to Barvinok~\cite{Bar}.

\begin{lemma} \label{lem:integral}
Fix a simplex $\sigma= {\rm conv}(x_0,x_1,\ldots,x_d)$ in $\RR^d$ and
an affine-linear function  $\ell: \RR^d \rightarrow \RR$, and
let $y_0=\ell(x_0), y_1 = \ell(x_1),\ldots,y_d = \ell(x_d)$
 be its values at the vertices. Then
$$ \int_{\sigma} {\rm exp}\bigl( \ell (t)\bigr) dt \,\, \,= \,
\,\,{\rm vol}(\sigma) \cdot \sum_{i=0}^d {\rm exp}(y_i) \! 
\!\prod_{j \in \{0,\ldots,d\}\backslash \{i\}} \!\! \!\!\!  (y_i-y_j)^{-1}. $$
\end{lemma}

\begin{proof}  This follows directly from equation (B.1) in \cite[Section B.1]{CSS}, 
and it can also easily be derived from Barvinok's formula in~\cite[Theorem 2.6]{Bar}. 
\end{proof}

This lemma implies the following formula for integrating exponentials
of piecewise-affine functions on a convex polytope.
This can be regarded as an exponential  variant of~(\ref{eq:intislinear}).

\begin{theorem}\label{thm:IntegralFormula}
Let $\Delta$ be a triangulation of the configuration $X=(x_1,\ldots,x_n)$
and $h : P \rightarrow \RR$ the piecewise-affine function
on $\Delta$ that takes values
$h(x_i) = y_i$ for $i = 1,2,\ldots,n$. Then
$$ \int_{P} {\rm exp}\bigl( h (t)\bigr) dt \,\, \,= \,\,\,
\sum_{i=1}^n
{\rm exp}(y_i) \sum_{\sigma \in \Delta: \atop i \in \sigma} \frac{{\rm vol}(\sigma)}
{\prod_{j \in \sigma \backslash i} (y_i-y_j)}
$$
\end{theorem}

\begin{proof}
We add the expressions in Lemma \ref{lem:integral} over all
maximal simplices $\sigma$ of the triangulation $\Delta$,
and we collect the rational function multipliers for 
 each of the $n$ exponentials ${\rm exp}(y_i)$.
 \end{proof}

This formula underlies the efficient solution to the estimation problem (\ref{eq:ourproblem})
that is implemented in the {\tt R} package {\tt LogConcDEAD} \cite{CGS}.
We record the following algebraic reformulation, which will be used in our
study in the subsequent sections. This follows from Theorem \ref{thm:IntegralFormula}.

\begin{corollary} \label{cor:optsecondary}
The equivalent optimization problems (\ref{eq:ourproblem}), 
(\ref{eq:constrained}), (\ref{eq:unconstrained}) are also equivalent to
\begin{equation}
\label{eq:optsecondary}
{\rm Maximize} \,\,\,w \cdot y \,- \,\sum_{\sigma \in \Delta}
\sum_{i \in \sigma}
\frac{{\rm vol}(\sigma) \cdot {\rm exp}(y_i)} 
{\prod_{j \in \sigma \backslash i} (y_i-y_j)},
\end{equation}
where $y $ runs over $ \RR^n$ and $\Delta$ is a regular triangulation of $X$ whose
secondary cone contains~$y$.
\end{corollary}

\begin{figure}[h]
\begin{center}\includegraphics[width = 0.65\textwidth]{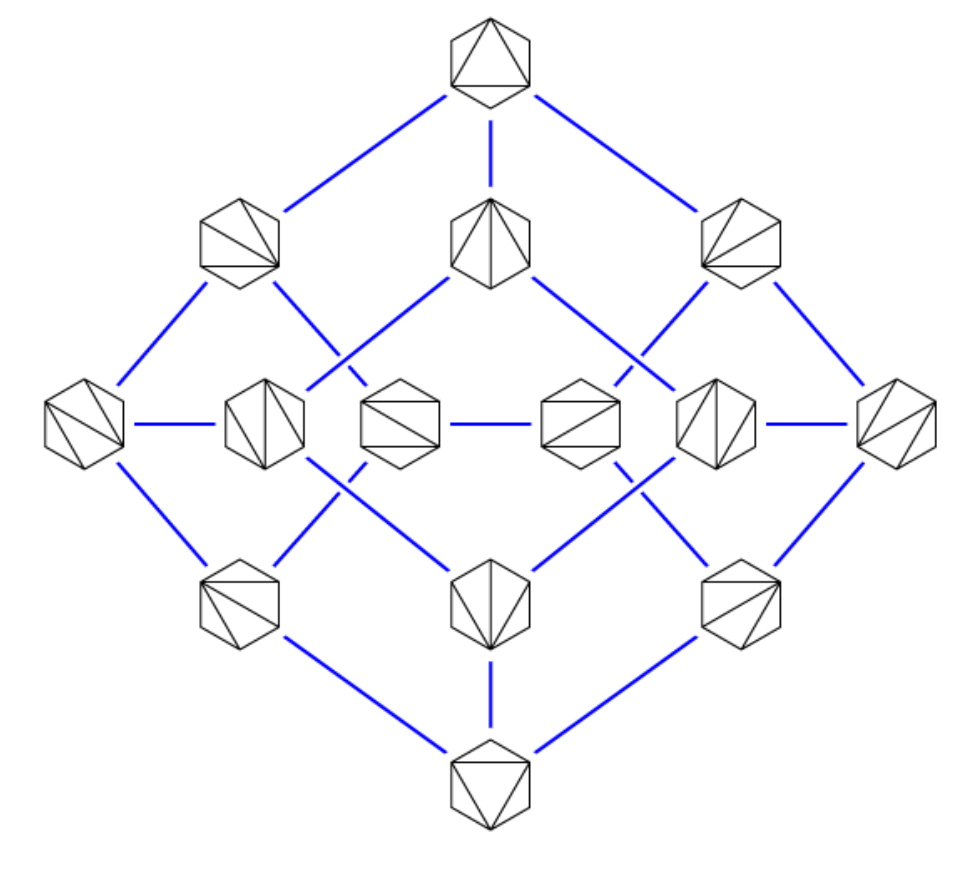}\end{center}
\vspace{-0.32in}
\caption{The associahedron is the secondary polytope for the six vertices of a hexagon.
This diagram belongs to David Epstein. It was first appeared in his blog post 
\cite{Epp1} and it was published as Figure 15 in his article
\cite{Epp2}.
\label{fig:associahedron}} 
\end{figure}

We close this section with an example that illustrates the various concepts seen so far.

\begin{example}
\label{ex:hexagon} \rm
Let $d=2$ and $n=6$.  Take $X$ to be six points in convex position in the plane,
labeled cyclically in counterclockwise order.
The normalized area of the triangle formed by any three of the vertices 
of the hexagon $P  = {\rm conv}(X)$ 
is computed as a $3 \times 3$-determinant
\begin{equation}
\label{eq:areas}
\quad v_{ijk} \,\, := \,\,
{\rm vol}\bigl( {\rm conv}(x_i, x_j, x_k) \bigr) \,\, = \,\,
{\rm det}\begin{pmatrix} 1 & 1 & 1 \\ x_i & x_j & x_k \end{pmatrix} 
\,\,\quad {\rm for}\, \,\, \,1 \leq i < j < k \leq 6. \quad
\end{equation}
The configuration $X$ has $14$ regular triangulations. These come in
three symmetry classes: six triangulations like
 $\Delta = \{123,134,145,156\}$, six triangulations like
 $\,\Delta' = \{123,134,146,456\}$, and two triangulations like
 $\,\Delta'' = \{123,135, 156,345\} $.
 The corresponding GKZ vectors are
$$ \begin{matrix}
z^{\Delta} &=& \bigl(\,v_{123}+v_{134}+v_{145}+v_{156}\,,\,
v_{123}\,,\,v_{123}+v_{134}\,,\,v_{134}+v_{145}\,,\,v_{145}+v_{156}\,,v_{156}\, \bigr)  ,\\
z^{\Delta'} &=&
\bigl(\,v_{123}+v_{134}+v_{146}\,,\, v_{123}\,,\, v_{123}+v_{134}\,,\, v_{134}+v_{146}+v_{456}\,,
\, v_{456}\,, \,v_{146}+v_{456}\, \bigr), \\
z^{\Delta'} &=&
\bigl(\, v_{123}+v_{135}+v_{156} \,,\, v_{123}\,, \, v_{123}+v_{135}+v_{345}\,,\,
v_{345}\,, \,v_{135}+v_{156}+v_{345}\,, \,v_{156}\, \bigr) ,
\end{matrix}
$$
as defined in (\ref{eq:GKZvector}).
The secondary polytope $\Sigma(X)$ is the convex hull of these $14$ points in $\RR^6$.
This is a simple $3$-polytope with $14$ vertices, $21$ edges and $9$ facets,
shown in Figure \ref{fig:associahedron}. This polytope is known as the {\em associahedron}. It  has $45 = 14+21+9+1$ faces in total, one for each of the $45$ polyhedral
subdivisions of $X$.  These are the supports of the functions $h_{X,y}$.

For example, the edge of $\Sigma(X)$ that connects $z^{\Delta}$ and $z^{\Delta'}$
represents the subdivision $\{123,134,1456\}$, with two triangles and one quadrangle.
The smallest face containing $\{z^{\Delta},z^{\Delta'},z^{\Delta''}\}$ is two-dimensional.
It is a pentagon, encoding the subdivision $\{123, 13456\}$.

The Samworth body
$\mathcal{S}(X)$ is full-dimensional in $\RR^6$. Its boundary
is stratified into $45$ pieces, one for each subdivision of $X$.
For any given $w \in \RR^6$, the optimal solution $y^*$ to~(\ref{eq:constrained})
lies in precisely one of these $45$ strata, depending on the shape
of the optimal density  $f_{X,y^*}$.

Algebraically, we can find $y^*$  by computing
the maximum among $14$ expressions like
\begin{equation}
\label{eq:maximumamong14}
 \begin{matrix} w_1 y_1 + w_2 y_2 +  \cdots + w_6 y_6
 & - & v_{123}\cdot \bigl( \frac{{\rm exp}(y_1)}{(y_1-y_2)(y_1-y_3)}+
 \frac{{\rm exp}(y_2)}{(y_2 - y_1)(y_2 - y_3)}+
 \frac{{\rm exp}(y_3)}{(y_3 - y_1)(y_3 - y_2)}\bigr) \smallskip \\
  & - & v_{134}\cdot \bigl( \frac{{\rm exp}(y_1)}{(y_1-y_3)(y_1-y_4)}+
 \frac{{\rm exp}(y_3)}{(y_3 - y_1)(y_3 - y_4)}+
 \frac{{\rm exp}(y_4)}{(y_4 - y_1)(y_4 - y_3)}\bigr) \smallskip \\
 & - & v_{145}\cdot \bigl( \frac{{\rm exp}(y_1)}{(y_1-y_4)(y_1-y_5)}+
 \frac{{\rm exp}(y_4)}{(y_4 - y_1)(y_4 - y_5)}+
 \frac{{\rm exp}(y_5)}{(y_5 - y_1)(y_5 - y_4)}\bigr)  \smallskip \\
 & - & v_{156}\cdot \bigl( \frac{{\rm exp}(y_1)}{(y_1-y_5)(y_1-y_6)}+
 \frac{{\rm exp}(y_5)}{(y_5 - y_1)(y_5 - y_6)}+
 \frac{{\rm exp}(y_6)}{(y_6 - y_1)(y_6 - y_5)}\bigr) .
\end{matrix}
\end{equation}
This formula is the objective function in
(\ref{eq:optsecondary}) for
 the triangulation $\Delta = \{123,134,145,156\}$.
 The mathematical properties of this optimization process will be studied in the next sections.
\end{example}

\section{Every Regular Subdivision Arises}

Our goal in this section is to prove Theorem~\ref{thm:converse}.
We begin by examining the function
\begin{equation}
\label{eq:defH}
H\,:\, \RR^d \rightarrow \RR\,,\,\,\,
(u_1,\ldots, u_d) \,\, \mapsto \,\, (-1)^d\frac{1 + u_1^{-1} + \cdots + u_d^{-1}}{u_1 u_2 \cdots u_d }
\,+\,  \sum_{j=1}^d \frac{e^{u_j}}{u_j^2 \prod_{k \not= j} (u_j - u_k)}.
\end{equation}


\begin{proposition}\label{prop:Hformula}
The function $H$ is well-defined on $\RR^d$. It admits the series expansion
\begin{equation}
\label{eq:Hidentity} H(u_1,\dots, u_d) \quad = \quad \sum_{r=0}^\infty\frac{h_r(u_1,\ldots,u_d)}{(r+d+1)!},
\end{equation}
where $h_r$ is the complete homogeneous symmetric polynomial of degree $r$
in $d$ unknowns.
\end{proposition}

\begin{proof}
We substitute the Taylor expansion of the exponential function in the sum on the right hand side of (\ref{eq:defH}).
This sum then becomes
$$\sum_{j=1}^d \frac{e^{u_j}}{u_j^2 \prod_{k \not= j} (u_j - u_k)} \quad= \quad\sum_{\ell=0}^\infty \frac{1}{\ell !} \sum_{j=1}^d \frac{u_j^{\ell-2}}{\prod_{k \not= j} (u_j - u_k)}
\quad
 $$
$$ = \quad
 \sum_{\ell=0}^\infty \frac{1}{\ell !} \sum_{j=1}^d \frac{u_j^{\ell-d-1}}{\prod_{k \not= j} (1-u_k/u_j )}
\quad = \quad
 \sum_{r=-d-1}^\infty \frac{1}{( r+d+1)!} \sum_{j=1}^d 
 \frac{u_j^r}{\prod_{k \not= j} (1-u_k/u_j )}
$$
For nonnegative values of the summation index $r = \ell - d- 1$, the inner summand 
equals $h_r(u_1,\ldots,u_d)$, by
Brion's Theorem \cite[Theorem~12.13]{MS}.
For negative values of $r$, we use Ehrhart Reciprocity,
in the form of \cite[Lemma 12.15, eqn~(12.7)]{MS},
as seen in \cite[Example 12.14]{MS}.
The two terms for $r \in \{-d-1,-d\}$
cancel with the left summand on the right hand side of (\ref{eq:defH}).
The terms for $r \in \{-d+1,\ldots,-2,-1\}$ are zero.
This implies~(\ref{eq:Hidentity}).
\end{proof}

We shall derive a useful integral representation of our function $H$.
What follows is a Lebesgue integral over the standard 
simplex $\,\Sigma_d = \{(y_1,\dots, y_d) \in \RR^d: \,y_i \geq 0 ,\, \sum_{i}y_i \leq 1\}$.

\begin{proposition}\label{prop:magic}
The function $H$ can be expressed as the following integral:
\begin{align}\label{eqn:HIntegralExpression}
H(u_1,\ldots, u_d) \,\,=\, \int_{\Sigma_d}
\left(1 - \sum_{i=1}^d t_i\right)\exp\left(\sum_{i=1}^d u_i t_i\right)\text{d} t_1\dots\text{d} t_d.
\end{align}
\end{proposition}

\begin{proof}
The complete homogeneous symmetric polynomial $h_r$ equals the Schur polynomial
$s_{(r)}$  corresponding to the partition $\lambda = (r)$. By formula (2.11) in \cite{GR} we have $s_{(r)} = Z_{(r)}$, where $Z_\lambda(u_1,\dots, u_d)$ is the {\em zonal polynomial}, or {\em spherical function} \cite{GR}. Therefore, we conclude 
$$H(u_1,\ldots, u_d) \,=\,\sum_{r=0}^{\infty}\frac{Z_{(r)}(u_1,\ldots, u_d)}{(r+d+1)!}
\,=\, \frac1{(d+1)!}\sum_{r=0}^{\infty}
\frac{Z_{(r)}(u_1,\ldots, u_d)\cdot [1]_{(r)}}{[d+2]_{(r)} \cdot r!},$$
where $[a]_\lambda = \prod_{j=1}^{m}(a-j+1)_{\lambda_j}$ for a partition $\lambda = (\lambda_1,\dots, \lambda_m)$, and $(a)_s = a(a+1)\cdots(a+s-1)$. In particular, $[1]_{(r)} = r!$, and $[1]_{\lambda} = 0$ if $\lambda$ has more than one nonzero part.  Therefore, 
$$ H(u_1,\ldots, u_d) \,\,= \,\, \frac1{(d+1)!}\sum_{\text{all partitions } \lambda} 
\frac{Z_\lambda(u_1,\ldots, u_d) \cdot [1]_\lambda}{[d+2]_\lambda \cdot |\lambda|!}.$$
By \cite[(4.14)]{GR}, this can be written in terms
of the confluent hypergeometric function of matrix argument ${}_1F_1$:
$$H(u_1,\ldots, u_d) \,\,= \,\,\frac1{(d+1)!} \cdot \, {}_1F_1(1;d+2;\diag(u_1,\ldots,u_d)) .$$
The right hand side has the desired integral representation (\ref{eqn:HIntegralExpression}),
by \cite[equation (5.14)]{GR}.
\end{proof}

\begin{corollary} \label{cor:positive}
 The function $H$ is positive, increasing in each variable, and convex.
\end{corollary}

\begin{proof}  The integrand in \eqref{eqn:HIntegralExpression} is
nonnegative. Hence, $H(u_1,\ldots, u_d) > 0$ for all $(u_1,\ldots, u_d)\in\mathbb R^d$. After taking derivatives with respect to $u_i$, the integrand remains positive.
 Therefore, $H$ is increasing in $u_i$. Finally, the integrand is a convex function,
and hence so is $H$.
\end{proof}

We now embark towards the proof of Theorem~\ref{thm:converse}.
Recall that a vector $y \in \RR^n$ is relevant if $h_{X,y}(x_i) = y_i$
for all $i$, i.e.~the regular subdivision of $X$ induced by $y$ uses each point $x_i$.

\begin{lemma}\label{lem:weightsAnyDim} Fix a configuration $X$ of $n$ points in $\mathbb R^d$.
For any relevant $y^* \in \RR^n$ that satisfies
$\int_{\RR^d} f_{X,y^*}(t)dt = 1$,  there are weights $w \in \RR^n_{> 0}$ 
such that $y^* $ is the optimal solution to \eqref{eq:ourproblem}-\eqref{eq:unconstrained}.
\end{lemma}

\begin{proof}
We use the formulation (\ref{eq:optsecondary}) which is equivalent to
 (\ref{eq:ourproblem}),  (\ref{eq:constrained}), and (\ref{eq:unconstrained}).
 Let $\Delta$ be any regular triangulation that refines the regular subdivision given by $y$.
 In other words, we choose $\Delta$ so that (\ref{eq:intislinear}) is maximized.
 The objective function in Corollary~\ref{cor:optsecondary} takes the form
 $$S(y_1,\dots, y_n) \,\,\,= \,\,\,w\cdot y - \sum_{i=1}^n\exp(y_i)\sum_{\sigma\in\Delta,\atop i\in\sigma}\frac{\text{vol}(\sigma)}{\prod_{j\in \sigma\setminus i }(y_i-y_j)}.$$

Consider the partial derivative of the objective function $S$ with respect to the unknown~$y_k$:
\begin{align*}
\frac{\partial S}{\partial y_k} \,\,\,= \,\,\,w_k \,\, - \,
&\sum_{\sigma\in\Delta,\atop k\in\sigma}\text{vol}(\sigma)\exp(y_k)\frac1{\prod_{j\in \sigma\setminus k }(y_k-y_j)}\left(1 - \sum_{j\in \sigma\setminus k }\frac{1}{(y_k - y_j)}\right) \\
-&\, \sum_{\sigma\in\Delta,\atop k\in\sigma}\text{vol}(\sigma)\sum_{j\in \sigma\setminus k } \exp(y_j) \frac1{\prod_{i\in \sigma\setminus j } (y_j - y_i)} \frac1{(y_j - y_k)}.
\end{align*}
Using the formula (\ref{eq:defH}) for the symmetric function $H(u_1,\ldots,u_d)$, this can be rewritten as
$$
\frac{\partial S}{\partial y_k} \,\,\,= \,\,\,w_k \,-\,
\sum_{\sigma\in\Delta,\atop k\in\sigma} \text{vol}(\sigma)\exp(y_k)H(\{ y_i - y_k  : 
i\in\sigma \backslash k\}).
$$
We now consider the specific given vector $y^* \in \RR^n$, and we use it to define
\begin{align}\label{weightsFormula}
w_k \,\,= \,\,\sum_{\sigma\in\Delta,\atop k\in\sigma} \text{vol}(\sigma)\exp(y^*_k)H(\{ 
y^*_i - y^*_k : i\in\sigma \backslash k \}).
\end{align}
 By Corollary \ref{cor:positive}, the vector $w=(w_1,\ldots,w_n)$
is well-defined and has positive coordinates. Consider now our
estimation problem (\ref{eq:ourproblem}) for that $w \in \RR^n_{>0}$.
By construction, the gradient vector of $S$ vanishes at $y^*$. Furthermore, recall that the choice
of the triangulation $\Delta$ was arbitrary, provided $\Delta$ refines the subdivision of $y$.
This ensures that  all subgradients of the objective function in
(\ref{eq:unconstrained}) vanish. Since this function is strictly convex,
as shown in \cite{CSS}, we conclude 
that the given  $y^*$ is the unique optimal solution for the choice of weights 
in (\ref{weightsFormula}).
\end{proof}

We note that the function $H$ and Lemma~\ref{lem:weightsAnyDim} are
quite interesting even in dimension one.

\begin{example} \label{ex:d=1} \rm Let $d=1$. So, we here examine
log-concave density estimation for $n$ samples $x_1 < x_2 < \cdots < x _n$ on the real line.
The function we defined in (\ref{eq:defH}) has the representations
$$ H(u) \,\,=\,\, \frac{e^u - u - 1}{u^2} \,\,=\,\,\int_{0}^1 (1-y) e^{uy} dy \,\,=\,\,
\frac{1}{2} + \frac{1}{6} u + \frac{1}{24} u^2 + \frac{1}{120} u^3 + \cdots .$$
A vector $y^* \in \RR^n$ is relevant if and only if
\begin{equation}
\label{eq:relevant}
{\rm det}
\begin{pmatrix}
1 & 1 & 1 \\
x_{i-1} & x_i  & x_{i+1} \\
y^*_{i-1} & y^*_i & y^*_{i+1}
\end{pmatrix} \,\leq \, 0 \,
\,\quad \hbox{for} \quad i=2,3,\ldots,n-1.
\end{equation}
The desired vector $w \in \RR^n_{>0}$ is defined by the formula in
(\ref{weightsFormula}). The $k$-th coordinate of $w$~is
$$w_k = \begin{cases} (x_2 - x_1)e^{y^*_1}H(y^*_2 - y^*_1)& {\rm if}\,\, k = 1,\\
(x_k - x_{k-1})e^{y^*_k}H(y^*_{k-1} - y^*_k) + (x_{k+1} - x_k)e^{y^*_k}H(y^*_{k+1} - y^*_k)&
{\rm if}\,\,   2\leq k \leq n-1,\\
(x_n-x_{n-1})e^{y^*_n}H(y^*_{n-1} - y^*_n)&  {\rm if}\,\, k = n.
\end{cases}$$
If we now further assume that $f_{X,y^*} = {\rm exp}(h_{X,y^*})$ is a density,
i.e.~$\int_{-\infty}^\infty f_{X,y^*}(t)dt = 1$, then $f_{X,y^*}$ is the
unique log-concave density that maximizes the likelihood function for $(X,w)$.
\end{example}

\begin{example} \label{ex:H2}
 \rm For $d=2$, our symmetric convex function $H$ has the form
$$H(u, v) \,=\,  \frac{1}{uv} + \frac{1}{u^2 v} + \frac{1}{u v^2}
+ \frac{e^u}{u^2(u{-}v)} + \frac{e^v}{v^2(v{-}u)} \,= \,
\frac{1}{6} + \frac{1}{24}(u+v) + \frac{1}{120}(u^2 + uv + v^2) + \cdots.
$$
For planar configurations $X$, we use this function to map 
each point $y^*$ in the boundary of
the Samworth body $\mathcal{S}(X)$ to a  hyperplane
$w \in \partial \mathcal{S}(X)^*$ that is tangent to $\partial \mathcal{S}(X)$ at $y^*$.
\end{example}

The set of all vectors $w \in \RR^n$ that lead to a desired optimal solution 
$y^* \in \partial \mathcal{S}(X)$
is a convex polyhedral cone in $\RR^n$. The following theorem characterizes 
that convex cone.

\begin{theorem} \label{thm:normalcone}
Fix a vector $y^* \in \RR^n$ that is relevant for $X$. Let
$\Delta_1,\Delta_2,\ldots, \Delta_m$ be all the regular triangulations of $X$
that refine the subdivision of $X$ given by $y^*$, and let
$w^{\Delta_i} \in \RR^n_{>0}$ be the vector defined 
by (\ref{weightsFormula}) for $\Delta_i$. Then, a vector $w \in \RR^n_{>0} $ lies in
the convex cone that is spanned by $\,w^{\Delta_1},w^{\Delta_2}, \ldots, w^{\Delta_m}\,$
if and only if $\,y^*$ is the optimal solution for 
(\ref{eq:ourproblem}),(\ref{eq:constrained}),(\ref{eq:unconstrained}),(\ref{eq:optsecondary}).
\end{theorem}

\begin{proof} This follows from the fact that the cone of subgradients at each $y^*$ is convex, and, the gradients for each triangulation on which $h_{X,y^*}$ is linear
are also subgradients at $y^*$; cf.~\cite{CSS}. We can take any convex combination of these subgradients to obtain another subgradient.
\end{proof}

\begin{example}[$n{=}4,d{=}2$]  \rm
\label{ex_4points_plane}
Fix four points $x_1,x_2, x_3, x_4$ in counterclockwise convex position in $\RR^2$.
These admit two regular triangulations, $\Delta_1 = \{124,234\}$ and
$\Delta_2 = \{123,134\}$. Consider any $y \in \RR^4$
with $\int_{\RR^2} f_{X,y}(t) dt = 1$.
 The vector
$w^{\Delta_1} \in \RR^4$ has coordinates 
\begin{align*}
w_1^{\Delta_1} &\,\,=\,\, v_{124}e^{y_1}H(y_2-y_1,y_4-y_1)\\
w_2^{\Delta_1} &\,\,=\,\, v_{124}e^{y_2}H(y_1-y_2, y_4-y_2)
       + v_{234}e^{y_2}H(y_3-y_2, y_4-y_2)\\
w_3^{\Delta_1} &\,\,=\,\, v_{234}e^{y_3}H(y_2-y_3, y_4-y_3)\\
w_4^{\Delta_1} &\,\,=\,\, v_{124}e^{y_4}H(y_1-y_4, y_2-y_4) 
+ v_{234}e^{y_4}H(y_2-y_4, y_3-y_4).
\end{align*}
Here $v_{ijk}$ denotes the triangle area in (\ref{eq:areas}).
Similarly, the vector $w^{\Delta_2}$ has coordinates
\begin{align*}
w_1^{\Delta_2} &\,\,=\,\, v_{123}e^{y_1}H(y_2-y_1, y_3-y_1)
 + v_{134}e^{y_1}H(y_3-y_1, y_4-y_1)\\
w_2^{\Delta_2} &\,\,=\,\, v_{123}e^{y_2}H(y_1-y_2, y_3-y_2)\\
w_3^{\Delta_2} &\,\,=\,\, v_{123}e^{y_3}H(y_1-y_3, y_2-y_3) 
+ v_{134}e^{y_3}H(y_1-y_3, y_4-y_3)\\
w_4^{\Delta_2} &\,\,=\,\, v_{134}e^{y_4}H(y_1-y_4, y_3-y_4).
\end{align*}
In these formulas, the bivariate function $H$ can be evaluated as in Example \ref{ex:H2}.

We now distinguish three cases for $y$, depending on the sign of the $4 \times 4$-determinant
\begin{equation}
\label{eq:tetrahedron}
{\rm det} \begin{pmatrix}
1 & 1 & 1 & 1 \\
x_1 & x_2 & x_3 & x_4 \\
y_1 & y_2 & y_3 & y_4
\end{pmatrix}.
\end{equation}
If (\ref{eq:tetrahedron}) is positive then
$y$ induces the triangulation $\Delta_1$.
In that case, $y$ is the unique
solution to our optimization problem whenever $w$ is any positive
multiple of $w^{\Delta_1}$.
If (\ref{eq:tetrahedron}) is negative then
$y$ induces $\Delta_2$ and it is the unique
solution whenever $w$ is a positive multiple of $w^{\Delta_2}$.
Finally, suppose (\ref{eq:tetrahedron}) is zero, so $y$
induces the trivial subdivision $1234$. If $w$ is any vector
in the cone spanned by $w^{\Delta_1}$ and $w^{\Delta_2}$ in $\RR^4$
then  $y$ is the optimal solution for 
(\ref{eq:ourproblem}),(\ref{eq:constrained}),(\ref{eq:unconstrained}),(\ref{eq:optsecondary}).
\end{example}

We next observe what happens in Theorem \ref{thm:normalcone}
when all coordinates of $y^*$ are equal.

\begin{corollary} \label{cor:cccc}
Fix the constant vector $y^* = (c,c,\ldots,c)$, where
 $c = - {\rm log}({\rm vol}(P))$, so as to ensure
 that $\int_{\RR^d} f_{X,y^*}(t) dt = 1$.
 For any regular triangulation~$\Delta_i$, the weight
vector in (\ref{weightsFormula}) is a constant multiple
of the GKZ vector in
(\ref{eq:GKZvector}). More precisely, we have $w^{\Delta_i} = \frac{e^c}{(d+1)!} \cdot z^{\Delta_i}$.
Hence $y^*$ is the optimal solution for any $w$ in the
cone over the secondary polytope $\Sigma(X)$.
\end{corollary}

\begin{proof} The constant term of the series expansion in
Proposition \ref{prop:Hformula} equals
$$ H(0,0,\ldots,0) \,\,= \,\, \frac{1}{(d+1)!} . $$
This implies that the sum in
(\ref{weightsFormula}) simplifies to 
$\frac{e^c}{(d+1)!}$
times the sum in (\ref{eq:GKZvector}).
The last statement follows from Theorem \ref{thm:normalcone}
because the cone over $\sigma(X)$ is spanned by
all GKZ vectors $z^\Delta$.
\end{proof}

We shall now prove the result that was stated in the Introduction.

\begin{proof}[Proof of Theorem~\ref{thm:converse}]
Let $\Delta_1, \ldots,\Delta_m$ be all regular triangulations
that refine a given subdivision $\Delta$.
To underscore the dependence on $y$, we write
$w^{\Delta_i}_y$ for the vector defined in (\ref{weightsFormula}).
Let $\mathcal{C}_\Delta$ denote the secondary cone of $\Delta$.
This is the normal cone to $\Sigma(X)$ at the face
with vertices $z^{\Delta_1},\ldots, z^{\Delta_m}$.
In particular, we have
$\,\dim (\text{span} (z^{\Delta_1},\dots, z^{\Delta_m})) = n - \dim (\mathcal{C}_\Delta)$.

For $y \in \RR^n$ we abbreviate $N(y) = 
\dim (\text{span} (w^{\Delta_1}_{y}, \dots, w^{\Delta_m}_{y}))$.
The closure of the cone $\mathcal{C}_\Delta$ contains 
the constant vector $y^0 = (c,c,\ldots,c)$,
where $c = - {\rm log}({\rm vol}(P))$.
Corollary \ref{cor:cccc} implies that  $N(y_0) = n - \dim (\mathcal{C}_\Delta)$.
The matrix $(w^{\Delta_1}_{y}, \dots, w^{\Delta_m}_{y})$
depends analytically on the parameter $y$. Its rank
 is an upper semicontinuous function of~$y$.
 Thus, there exists an open ball $\hat{\mathcal{B}}$ in $\RR^n$
that contains $y_0$ and such that
 $N(y)\geq n - \dim (\mathcal{C}_\Delta)$ for every $y\in\hat{\mathcal B}$. 
  Now, let $\mathcal B = \mathcal{C}_\Delta \cap \hat{\mathcal B}$. 
  The set $\mathcal B$ is full-dimensional in $\mathcal{C}_\Delta$,  and 
  $N(y)\geq n  - \dim(\mathcal{C}_\Delta)$ for all $y \in \mathcal{B}$.

For each $ y \in \mathcal{B}$ we consider the convex cone 
in Theorem \ref{thm:normalcone}, which consists of all weight
vectors $w$ for which the optimum occurs at $y$.  We denote it by
$\,{\rm cone}(w^{\Delta_1}_{y}, \ldots, w^{\Delta_m}_{y})$.
These convex cones are pairwise disjoint as $y$ runs over $\mathcal{B}$,
and they depend analytically on $y$. Since the dimension
of each cone is at least $n-{\rm dim}(\mathcal{B}) $, it follows that the 
semi-analytic set
\begin{equation}
\label{eq:fulldimset}
 \bigcup_{y \in \mathcal{B}} {\rm cone}(w^{\Delta_1}_{y}, \dots, w^{\Delta_m}_{y}). 
 \end{equation}
is full-dimensional in $\RR^n$.
By Theorem \ref{thm:normalcone}, for each $w$ in the set (\ref{eq:fulldimset}),
the optimal solution 
$\hat f$ to (\ref{eq:ourproblem}) is a piecewise log-linear function whose
regions of linearity are the cells of~$\Delta$.
\end{proof}

We conclude this section with the following open problem.

\begin{problem} \label{prob:dim}

Show that the rank $N(y)$ of the matrix $(w^{\Delta_1}_{y}, \ldots, w^{\Delta_m}_{y})$ is the same
for all vectors $y$ that induce the regular subdivision $\Delta$, namely we have
 $\,N(y) = n-{\rm dim}(\mathcal{C}_\Delta)$.
\end{problem}

For the proof of Theorem~\ref{thm:converse}, it was sufficient
to only have this constant-dimension property for all $y$ in a
relatively open subset $\mathcal{B}$ of
the secondary cone $\mathcal{C}_\Delta$.
Problem \ref{prob:dim} puts forth the conjecture that this property
holds throughout the entire secondary cone  $\mathcal{C}_\Delta$.

\section{The Samworth Body}

The maximum likelihood problem studied in this paper is
a linear optimization problem over a convex set. We named that convex set the Samworth body, in
recognition of the contributions made by Richard Samworth and his collaborators 
\cite{CGS, CSS}. In what follows we explore the geometry of the Samworth body.
We begin with the following explicit formula:

\begin{corollary}
The Samworth body of a given configuration $X$ of $\,n$ points in $\RR^d$ equals
\begin{equation} \label{eq:samformula}
\! \mathcal{S}(X) \,\, = \,\,
\biggl\{ (y_1,\ldots,y_n) \in \RR^n \,:\,
\sum_{\sigma \in \Delta}
\sum_{i \in \sigma}
\frac{{\rm vol}(\sigma) \cdot {\rm exp}(y_i)} 
{\prod_{j \in \sigma \backslash i} (y_i-y_j)} \leq 1  \,\,\,
\hbox{for all $\Delta$ that refine $y\,$} \biggr\}.
\end{equation}
This is a closed convex subset of $\,\RR^n$.
In the defining condition
we mean that  $\Delta$ runs over all regular triangulations that 
refine the regular polyhedral subdivision of $X$ specified by $y$.
\end{corollary}

\begin{proof}
This is a reformulation of the definition (\ref{eq:samdef}) using
the formulas in Theorem \ref{thm:IntegralFormula}
and Corollary \ref{cor:optsecondary}. Closedness
and strict convexity of $\mathcal{S}(X)$ were noted  in Theorem \ref{thm:samworth}.
\end{proof}

Maximization of a linear function $w$ over $\mathcal{S}(X)$
becomes an unconstrained problem via the
Legendre-Fenchel transform as in (\ref{eq:optsecondary}).
By solving this problem for many instances of~$w$,
one can approximate the shape of $\mathcal{S}(X)$.
Indeed, each regular subdivision of $X$ specifies a full-dimensional
subset in the boundary of the dual body $\mathcal{S}(X)^*$, by
Theorem~\ref{thm:converse}. If we choose a direction $w$ at random
in $\RR^n$, then a unique positive multiple $\lambda w$ lies
in $\partial \mathcal{S}(X)^*$, in the stratum associated to
the subdivision of $X$ specified by
the optimal solution $y^* \in \partial \mathcal{S}(X)$.
By evaluating the map $w \mapsto y^*$ many times,
we thus obtain the empirical distribution on the subdivisions,
indicating the proportion of volumes of the strata in  $\partial \mathcal{S}(X)^*$.
In the next example we compute this distribution when the
double sum in (\ref{eq:samformula}) looks like that in
(\ref{eq:maximumamong14}).

\begin{example} \label{ex:associahedron2} \rm
Let $d=2$, $n=6$, and take our configuration $X$
to be the six points $(0,0),(1,0),(2,1),(2,2),(1,2), (0,1)$.
We sampled 100,000 vectors $w$ uniformly from the
simplex $\{w \in \RR^6_{\geq 0} : \sum_{i=1}^6 w_i = 1\}$.
For each $w$, we computed the optimal $y^* \in \RR^6$,
and we recorded the subdivision of $X$ that is the support of $h_{X,y^*}$.
We know from  Example \ref{ex:hexagon} that the secondary polytope
$\Sigma(X)$ is an associahedron, which has $14+21+9+1  = 45$ faces.
We here code each subdivision by a list of length $3,2,1$ or $0$ from  among the diagonal
segments
$$ 13, \,14, \, 15, \,24, \,25,\,26, \,35, \,36,\,46. $$
For instance, the list $13 \,\,14 \,\,15$ encodes the triangulation $\Delta$ in Example~\ref{ex:hexagon}.
The edge connecting the triangulations $\Delta$ and $\Delta'$ from Example~\ref{ex:hexagon} is denoted  $13 \, 14 $.
We write $\emptyset$ for the trivial flat subdivision.
The following table of percentages shows the empirical distribution 
we observed for the $45$ outcomes of our experiment:

\setcounter{MaxMatrixCols}{20}
\begin{small}
$$ \begin{matrix}
\emptyset & 35 \,&\,   46 \,&\,   24 \,&\, 15 \,&\,   13 \,&\,  26  \,&\, 25 \,&\, 14 \,&\, 36  \\
30.5 \,\,&\,\, 5.95 \,\,&\,\,  5.85 \,\,&\,\,  5.84 \,\,&\,\, 5.83 \,\,&
\,\, 5.75 \,\, &\,\, 5.70 \,\,&\,\, 3.91 \,\,&\,\, 3.90 \,\,&\,\, 3.87  \\
\end{matrix}
$$
$$
 \begin{matrix}
 13 \,15 &   26 \, 46 &  15 \, 35 & 13 \, 35 & 24 \, 26 &  24 \, 46 &  13 \, 14 & 35\, 36 &  14 \, 24  & 26 \, 36 & 14 \, 46 &  25 \, 35  & 15 \, 25 \\
 1.23 & 1.21  & 1.21 & 1.20 & 1.16 & 
1.14 & 0.96 & 0.92  & 0.92 & 0.92 & 0.92 &  0.90 & 0.90
 \end{matrix}
 $$
 $$ \begin{matrix}
 25 \,26 & 
 14 \, 15 &
 36 \, 46 & 
 24 \, 25 &
 13 \, 36 &
 13 \,46 &
 26 \, 35 &
 15\, 24 &
 13 \, 14 \, 15 &
  13 \,15 \,35 &
 14 \,24 \,46 &
 24 \,26 \, 46 \\
0.89 &
0.89 &
0.87 &
0.87 &
0.84 &
0.82 &
0.77 &
0.70 &
0.25 &
0.24 & 0.23 & 0.22
 \end{matrix}
 $$
  $$
 \begin{matrix}
  15 \, 25 \, 35 &
 26\, 36 \, 46 &
 13 \, 35 \, 36 &
 24 \, 25 \, 26 &
 13 \, 36 \, 46 &
 25 \, 26\,  35 &
 15 \,24 \, 25 &
 14 \, 15 \, 24 &
 13 \,14 \, 46 &
 26 \, 35\,  36 \\
0.22 &
0.21 &
0.20 &
0.18 &
0.18 &
0.16 &
0.15 &
0.15 & 0.15 & 0.14 
 \end{matrix}
 $$
\end{small}
The entry marked $\emptyset$ reveals that the trivial subdivision
occurs with the highest frequency. This means that
a large portion of the dual boundary $\partial \mathcal{S}(X)^*$ is  flat.
Equivalently, the Samworth body $\mathcal{S}(X)$ has a ``very sharp edge'' along the lineality
space of the secondary~fan.
\end{example}

To get a better understanding of the geometry of the
Samworth body $\mathcal{S}(X)$, at least when $d$ or $n-d$ are small, we 
can also use the algebraic formula in (\ref{eq:samformula})
for explicit computations.

\begin{example} \label{ex:fancyteewurst} \rm
Let $d=3$, $n=6$, and fix the configuration of vertices of a {\em regular octahedron}:
$$ X \, = \,(x_1,x_2,\ldots,x_6) \,=\, \bigl( \,+e_1\,,\,-e_1\,,\,\,+e_2,\,-e_2\,,\,\,+e_3\,,\,-e_3 \,\bigr). $$
Here $e_i$ denotes the $i$th unit vector in $\RR^3$. The secondary polytope
$\Sigma(X)$ is a triangle.  Its edges correspond to the three subdivisions of the
 octahedron  $X$ into
two square-based pyramids,
$\Delta_{1234} = \{12345, 12346\}$,
$\Delta_{1256} = \{12356,12456\}$, and
$\Delta_{3456} = \{13456, 23456\}$.
Its vertices correspond to the three triangulations of $X$, namely
$\Delta_{12} = \{1235, 1236, 1245, 1256\}$,
$\Delta_{34} = \{ 1345, 1346, 2345, 2346 \}$,
and $\Delta_{56} = \{1356, 1456, 2356, 2456\}$.

The normal fan of $\Sigma(X)$, which is the secondary fan of $X$,
has three full-dimensional cones in $\RR^6$.
A vector $y$ in $\RR^6$ selects the triangulation
$\Delta_{ij}$ if $y_i+y_j$ is the uniquely attained
minimum among $\{y_1+y_2,\,y_3+y_4,\,y_5+y_6\}$.
It selects  $\Delta_{1234}$ if $y_1+y_2 = y_3+y_4 < y_5+y_6$,
and it leaves the octahedron unsubdivided when $y$ is
in the lineality space $\,\{y \in \RR^6: y_1+y_2 = y_3+y_4 = y_5+y_6\}$.

The Samworth body $\mathcal{S}(X)$ is defined 
in $\RR^6$ by the following system of three inequalities.
Use the $i$th inequality when the $i$th number in the list
$(y_1{+}y_2, y_3 {+} y_4, y_5 {+} y_6)$ is the smallest:
$$
\begin{matrix}
 \frac{ e^{y_1} (2 y_1-y_6-y_5) (2 y_1-y_4-y_3)}{(y_1{-}y_2)(y_1{-}y_3)(y_1{-}y_5)(y_1{-}y_6)(y_1{-}y_4) }
- \frac{ e^{y_2} (2 y_2-y_6-y_5) (2 y_2-y_4-y_3)}{(y_1{-}y_2)(y_2{-}y_3)(y_2{-}y_5)(y_2{-}y_6)(y_2{-}y_4) }
+ \frac{e^{y_3} (2 y_3-y_6-y_5) }{ (y_1{-}y_3)(y_2{-}y_3)(y_3{-}y_5)(y_3{-}y_6) } \smallskip \\ \,\,
+ \frac{e^{y_4} (2 y_4-y_6-y_5) }{ (y_1{-}y_4)(y_2{-}y_4)(y_4{-}y_5)(y_4{-}y_6) }
- \frac{e^{y_5} (y_4-2 y_5+y_3) }{ (y_1{-}y_5)(y_2{-}y_5)(y_3{-}y_5)(y_4{-}y_5) }
- \frac{e^{y_6} (y_4-2 y_6+y_3) }{ (y_1{-}y_6)(y_2{-}y_6)(y_3{-}y_6)(y_4{-}y_6) } \,\,\leq \,1 \smallskip
\end{matrix} 
$$
$$
\begin{matrix}
  \frac{ e^{y_1} (2 y_1-y_6-y_5) }{ (y_1-y_3)(y_1-y_4)(y_1-y_5)(y_1-y_6) }
+ \frac{ e^{y_2} (2 y_2-y_6-y_5) }{ (y_2-y_3)(y_2-y_4)(y_2-y_5)(y_2-y_6) }
- \frac{ e^{y_3} (2 y_3-y_6-y_5) (y_2-2 y_3+y_1)}{(y_1-y_3)(y_3-y_4)(y_3-y_5)(y_3-y_6)(y_2-y_3) }
\smallskip \\
+ \frac{ e^{y_4} (2 y_4-y_6-y_5) (-2 y_4+y_1+y_2)}{(y_1-y_4)(y_3-y_4)(y_4-y_5)(y_4-y_6)(y_2-y_4) }
- \frac{ e^{y_5} (y_2-2 y_5+y_1) }{ (y_1-y_5)(y_2-y_5)(y_3-y_5)(y_4-y_5) }
- \frac{ e^{y_6} (y_2-2 y_6+y_1) }{ (y_1-y_6)(y_2-y_6)(y_3-y_6)(y_4-y_6) }\,\,\leq \,1 \smallskip
\end{matrix}
$$
$$ \!
\begin{matrix}
\frac{ e^{y_1} (2 y_1-y_4-y_3) }{ (y_1-y_3)(y_1-y_4)(y_1-y_5)(y_1-y_6) }
+ \frac{ e^{y_2} (2 y_2-y_4-y_3) }{ (y_2-y_3)(y_2-y_4)(y_2-y_5)(y_2-y_6) }
- \frac{ e^{y_3} (y_2-2 y_3+y_1) }{ (y_1-y_3)(y_2-y_3)(y_3-y_5)(y_3-y_6) } - \smallskip \\
 \frac{ e^{y_4} (-2 y_4+y_1+y_2) }{ (y_1-y_4)(y_2-y_4)(y_4-y_5)(y_4-y_6) }
{+} \frac{ e^{y_5} (y_4-2 y_5+y_3)   (y_2-2 y_5+y_1)}{(y_1-y_5)(y_3-y_5)(y_5-y_6)(y_4-y_5)(y_2-y_5) }
{-} \frac{ e^{y_6} (y_4-2 y_6+y_3)  (y_2-2 y_6+y_1)}{(y_1-y_6)(y_3-y_6)(y_5-y_6)(y_4-y_6)(y_2-y_6) }
\leq 1
\smallskip
\end{matrix}
$$

The dual convex body $\mathcal{S}(X)^*$ has 
seven strata of faces in its boundary: a $3$-dimensional manifold of $2$-dimensional faces,
corresponding to the trivial subdivision,  three $4$-dimensional manifolds of edges
 corresponding to $\Delta_{1234}, \Delta_{1256}, \Delta_{3456}$,
and three $5$-dimensional manifolds of extreme points, 
corresponding to $\Delta_{12},\Delta_{34},\Delta_{56}$.
Each $2$-dimensional face of $\mathcal{S}(X)^*$ is a triangle, 
like the secondary polytope $\Sigma(X)$.
The dual to this convex set is the Samworth body $\mathcal{S}(X)$, which is
strictly convex. Its boundary is singular along three $4$-dimensional 
strata are formed when two of the three inequalities above are active.
These meet in a highly singular $3$-dimensional stratum which is
formed when all three inequalities are active. 
These singularities of $\partial \mathcal{S}(X)$ exhibit the secondary fan of $X$.
It is instructive to draw a cartoon, in dimension two or three, to visualize the
boundary features of $\mathcal{S}(X)$ and $\mathcal{S}(X)^*$.
\end{example}

\smallskip

Up until this point, the premise of this paper has been that
the configuration $X$ is fixed but the weights $w$ vary.
Example \ref{ex:fancyteewurst} was meant to give an impression of
the corresponding geometry, by describing in an intuitive language how 
a  Samworth body  $\mathcal{S}(X)$ can look like.

However, our premise is at odds with the
perspective of statistics. For a statistician,
the natural setting is to fix unit weights,
$w = \frac{1}{n}(1,1,\ldots,1)$,
and to assume that $X$ consists of $n$ points
that have been sampled from some underlying distribution.
Here, one cares about one distinguished point in
$\partial \mathcal{S}(X)$ and less about the
global geometry of the Samworth body.
Specifically, we  wish to know which
face of $\mathcal{S}(X)^*$ is pierced by the  ray
$\bigl\{ (\lambda,\ldots,\lambda) \,:\, \lambda \geq 0 \bigr\}$.

 \begin{table}[!b]
 \begin{center}
\begin{tabular}{| c c c c | c | c c c c |}
\hline\multicolumn{4}{|c|}{Subdivision: number of} & Convex & Gaussian & Uniform & Circular & Circular \\ 
3-gons & 4-gons & 5-gons & 6-gons & hull & $\mathcal{N}(0,1)$ & $a=0.5$ & $a=0.3$ & $a=0.1$ \\ \hline 
1 &  0 &  0 &  0 &  3&  948 & 533 & 257 & 34 \\
0 &  1   &0  & 0         &4      &8781 & 6719 & 4596 & 1507 \\
0   &0   &1   &0         &5     & 8209 & 9743 & 10554 & 8504 \\
0   &0   &0   &1         &6    &  1475 & 2805 & 4495 & 9887 \\
2   &0   &0   &0         &4      &  8& 3 & 6 & 7 \\
1   &1   &0   &0         &5     &   1 & 2 & 1 & 2 \\
3   &0   &0   &0         &3    &    6 & 2 &  2 & 1 \\
2  & 1   &0   &0         &4   &    39 & 16& 4 & 7 \\
2   &0   &1   &0         &5        &1 & 1 & 0 & 1 \\
1   &2   &0   &0         &5       & 1 & 0 & 1 & 6 \\
4  & 0   &0   &0         &4      &  1 & 0 & 0 & 0 \\
3   &1   &0   &0         &3     & 114 & 38 & 10 & 1 \\ 
3   &0   &1   &0         &4       &39 & 20 & 9 & 2 \\
2   &2   &0   &0         &4      & 59 & 19 & 16 & 9 \\
5   &0   &0   &0         &3     &   3 & 0 & 0 & 0 \\
4   &1   &0   &0         &4       & 1 & 0 & 0 & 0 \\
4   &0   &1   &0         &3      & 90 & 27 & 8 & 1 \\
3   &2   &0   &0         &3     & 120 & 32 & 11 & 0 \\
5   &1   &0   &0         &3    &   50 & 11 & 3 & 0 \\
7   &0   &0   &0         &3    &    2 & 1 & 0 & 0 \\ \hline
\end{tabular}
\end{center}
\caption{\label{tab:caroline} The optimal subdivisions
for six random points in the plane}
\end{table}

\begin{example} \label{ex:Gaussian} \rm
Let $d=2$ and $n=6$ as in Example \ref{ex:associahedron2},
but now with unit weights $w = \frac{1}{6}(1,1,1,1,1,1)$.
We  sample  i.i.d.~points $x_1,\ldots,x_6$ from various distributions $f$ on $\RR^2$,
some log-concave and others not,
and we compare the resulting maximum likelihood densities~$\hat f$.

In what follows, we analyze the case where $f$ is a standard Gaussian distribution or a uniform distribution on the unit disc, and we contrast this to distributions of the form $X=(U_1^a \cos(2\pi U_2), U_1^a \sin(2\pi U_2))$, where $U_1$ and $U_2$ are independent uniformly distributed on the interval $[0,1]$ and $a<0.5$. Such distributions have more mass towards the exterior of the unit disc and are hence not log-concave. For $a=0.5$ this is the uniform distribution on the unit disc. We drew 20,000
samples $X = (x_1,\ldots,x_6)$ from each of these four distributions.

 For each experiment, we recorded the number of vertices of the convex hull of the sample,
we computed the optimal subdivision using {\tt LogConcDEAD},
 and we recorded the shapes of its cells.  Our results are reported in Table \ref{tab:caroline}.
    Each of the four right-most columns shows the number of experiments out of 20,000 that
resulted in a  subdivision as described in the five left-most columns.  
These columns do not add up to 20,000, because we discarded all experiments for which the optimization procedure did not converge due to numerical instabilities.
 
In the vast majority of cases, reported in the first four rows,
the optimal solution $\hat f$ is log-linear. Here the
subdivision is trivial, with only one cell. For instance, the fourth row
is the 30.5\% case in  Example \ref{ex:associahedron2}. In  the last row,
 ${\rm conv}(X)$ is a triangle
and the subdivision is a triangulation that uses all three interior points.
We saw such a triangulation in Example~\ref{ex:octahedron}.
In fact, we constructed the data (\ref{eq:sixpoints})
by modifying one of the examples with seven cells found by sampling from a Gaussian $\mathcal{N}(0,1)$ distribution. Note that the subdivisions resulting from Gaussian samples tend to have more cells than
those from other distributions.
\end{example}

The examples in this section illustrate two different interpretations
of the data set $(X,w)$: either 
the configuration $X$ is fixed and the weight vector $w$ varies, or $w$ is fixed and $X$ varies.
These are two different parametric versions of our optimization problem (\ref{eq:ourproblem}),
(\ref{eq:constrained}),  (\ref{eq:unconstrained}), (\ref{eq:optsecondary}).
This generalizes the interpretation of the
secondary polytope $\Sigma(X)$
seen in \cite[Section 1.2]{DRS}, namely as a   geometric
model for {\em parametric linear programming}.
The vertices of $\Sigma(X)$ represent the various collections of optimal bases 
when the matrix $X$ is fixed and the cost function $w$ varies.
See \cite[Exercise 1.17]{DRS}  for the case $d=2, n=6$, as in
 Examples \ref{ex:hexagon}, \ref{ex:associahedron2} and \ref{ex:Gaussian}.
Of course, it is very interesting to examine what happens when both $X$ and $w$
vary, and to study $\Sigma(X)$ as a function on the space of configurations $X$.
This was done in  \cite{universal}.
The same problem is even more intriguing in the statistical setting
introduced in this paper.

\begin{problem} 
Study the Samworth body as a function $X \mapsto \mathcal{S}(X)$ on
the space of configurations. Understand
log-concave density estimation as a parametric optimization problem.
\end{problem}

This problem has many angles, aspects and subproblems. Here is one of them:

\begin{problem}
For fixed $w$ and a fixed combinatorial type of subdivision $\Delta$,
study the semi-analytic set of all configurations $X$ such that 
$\Delta$ is the optimal subdivision for the data $(X,w)$.
\end{problem}

For instance, suppose we fix the triangulation $\Delta$ seen on the right of 
Figure~\ref{fig:octahedron}. How much can we perturb the configuration
in (\ref{eq:sixpoints}) and retain that $\Delta$ is optimal for unit weights?
For $n{=}6,d{=}2$, give inequalities that characterize
the space of all datasets $(X,w)$ that select $\Delta$.

An ultimate goal of our geometric approach is the design of new
tools for nonparametric statistics.
One aim is the development of test statistics for assessing
whether a given sample comes from a log-concave distribution.
Such tests are important, e.g.~in economics \cite{An1,An2}.

\begin{problem}
Improve the accuracy of existing test statistics for log-concavityy~\cite{Chen_Samworth, Hazelton}
by augmenting these with combinatorial properties
(such as the f-vector) of the observed subdivision $\Delta$.
\end{problem}

The idea is that $\Delta$ is likely to have more cells
when $X$ is sampled from a log-concave distribution.
Hence we might use the f-vector of $\Delta$ as
a test statistic for log-concavity. 
The study of such tests seems related to the
approximation theory of convex bodies developed
by Adiprasito, Nevo and Samper \cite{ANS}.
What does their
``higher chordality'' mean for statistics?

\section{Unit Weights}

In this section we offer an analysis of the uniform weights case.
Example \ref{ex:Gaussian} suggests that the  flat subdivision occurs with overwhelming
probability when the sample size is small. Our main result in this section establishes this flatness
for the smallest non-trivial case $n=d+2$:

\begin{theorem}\label{thm:d+2points}  
Let $X$ be a configuration of $n=d+2$ points that affinely span $\RR^d$. For
$w = \frac{1}{n}(1,\ldots,1)$, the optimal density
$\hat f$ is log-linear, so
the optimal subdivision of $X$ is~trivial.
\end{theorem}

We shall use the following lemma, which can be derived by a direct computation.

\begin{lemma}\label{lem:surprise} The symmetric function $H$ in 
Section 4 satisfies the differential equation
$$\frac{\partial H}{\partial x_1}
(x_1,\dots, x_d) \quad = \quad \frac{e^{x_1}H(-x_1, x_2-x_1,\dots, x_d-x_1) - H(x_1,\dots, x_d)}{x_1}.$$
\end{lemma}

\begin{proof}[Proof of Theorem~\ref{thm:d+2points}] 
Our $d+2$ points in $\mathbb R^d$ can be partitioned uniquely into two 
affinely independent subsets whose
convex hulls intersect. This gives rise to a unique identity
$$\sum_{i=1}^k \alpha_ix_i \,\,\,= \,\, \sum_{j=k+1}^{d+2}\beta_j x_j,$$
where $1\leq k \leq d+1, \,\,\alpha_1,\dots, \alpha_k, \beta_{k+1},\dots, \beta_{d+2} \geq 0$, and $\sum \alpha_i = \sum \beta_j = 1$. We abbreviate $\mathcal{D} = \{1,2,\ldots,d+2\}$.
There are precisely three regular subdivisions of the configuration $X$:
\begin{enumerate}
\item[(i)] the triangulation $\,\bigl\{ \mathcal{D} \backslash  \{1\}, 
\mathcal{D} \backslash  \{2\}, \ldots, \mathcal{D} \backslash  \{k\} \bigr\}$,
\item[(ii)] the triangulation
$\,\bigl\{ \mathcal{D} \backslash  \{k{+}1\}, 
\mathcal{D} \backslash  \{k{+}2\}, \ldots, \mathcal{D} \backslash  \{d{+}2\} \bigr\}$,
\item[(iii)] the flat subdivision $\bigl\{ \mathcal{D} \bigr\}$.
\end{enumerate}
The simplex volumes $\,\sigma_{\mathcal{D} \setminus i} 
= {\rm vol}\bigl({\rm conv}(\,x_\ell : \ell \in \mathcal{D} \backslash \{i\})\bigr) \,$ satisfy the identity
\begin{equation}
\label{eq:simplexvolumes}
\qquad  \quad \sum_{i=1}^k \sigma_{\mathcal{D}\setminus i} \,\,\,= \,\,\sum_{j=k+1}^{d+2} 
 \! \sigma_{\mathcal{D}\setminus j}
  \quad = \quad {\rm vol}({\rm conv}(X)).   
\end{equation}

Now let $w \in \RR^{d+2}$ be a positive weight vector, and
suppose that the optimal heights $y_1,\dots, y_{d+2}$ do not induce the flat subdivision (iii). 
This means that the optimal subdivision is one  of the triangulations (i) and (ii).
We will show that in that case $w \not=  (\lambda, \lambda,\ldots,\lambda)$.

After relabeling we may assume that (ii) is the optimal
triangulation for the given weights~$w$.
This is equivalent to the inequality
$$\sum_{i=1}^k y_i\sigma_{\mathcal{D}\setminus i} 
\,\,\,> \,\, \sum_{j=k+1}^{d+2}y_j \sigma_{\mathcal{D}\setminus j}. \quad \qquad $$
In light of (\ref{eq:simplexvolumes}),
 at least one of $y_1,\dots, y_{k}$ has to be larger than at least one of $y_{k+1},\dots, y_{d+2}$. 
 After relabeling once more, we may assume that  $\,y_1 > y_{k+1}$.

Theorem \ref{thm:normalcone} states that the weight vector $w$ is
uniquely determined (up to scaling) by the optimal height vector $y$.
Namely, the coordinates of $w$ are given by the formula \eqref{weightsFormula} for
the optimal triangulation (ii). That formula gives
\begin{equation}
\label{eq:w_1}
w_1 \,\,\,= \,\, \sum_{j=k+1}^{d+2} \sigma_{\mathcal{D}\backslash j}e^{y_1}
H \bigl(y_\ell-y_1 : \ell\in\mathcal{D}\backslash \{1, j\} \bigr),  \quad
\end{equation}
and
\begin{equation}
\label{eq:w_{k+1}}
w_{k+1} \,\,\,=\,\, \sum_{j=k+2}^{d+2}\sigma_{\mathcal{D}\backslash j}e^{y_{k+1}}
H \bigl(y_\ell-y_{k+1}: \ell\in \mathcal{D}\backslash \{k{+}1,j\} \bigr).
\end{equation}
For any index $j\in \{k{+}2,\dots, d{+}2\}$ we consider the expression
\begin{align}\label{eq:difference}
e^{y_{1}}H(y_\ell-y_1 : \ell\in\mathcal{D}\backslash j )\,\,-\,\,
e^{y_{k+1}}H(y_\ell-y_{k+1}: \ell\in \mathcal{D}\backslash j)\, \, \\
= \,\,\, \bigl(\,e^{y_1-y_{k+1}} H(y_{\ell} - y_{k+1} - (y_1-y_{k+1}) : \ell\in \mathcal{D}\backslash j)
\, - \,H(y_\ell-y_{k+1}: \ell\in \mathcal{D}\backslash j)
\,\bigr). \notag
\end{align}
If we divide the parenthesized difference by $x_1 = y_1 - y_{k+1}$,
then we obtain an expression as in the right hand side of Lemma~\ref{lem:surprise}.
Then, by Lemma~\ref{lem:surprise}, the expression in~\eqref{eq:difference} becomes
$$ e^{y_{k+1}} \cdot (y_1-y_{k+1}) \cdot \frac{\partial H}{\partial x_1}\bigl( \,
y_\ell-y_{k+1}: \ell\in \mathcal{D}\backslash j \,\bigr) .$$
By Corollary \ref{cor:positive}, all partial derivatives of $H$ are positive.
Also, recall that $y_1>y_{k+1}$. Therefore, the expression in~\eqref{eq:difference} is positive. Hence, for any $j\in\{k{+}2,\dots, d{+}2\}$,
we have
$$ e^{y_{1}}H(y_\ell-y_1 : \ell\in\mathcal{D}\backslash j)
\,\,\,> \,\,\, e^{y_{k+1}}H(y_\ell-y_{k+1}: \ell\in \mathcal{D}\backslash j). $$
In the left expression it suffices to take $ \,\ell\in\mathcal{D}\backslash \{1, j\} $,
and in the right expression it suffices to take $\, \ell\in \mathcal{D}\backslash \{k{+}1,j\}$.
Summing over all $j$, the identities (\ref{eq:simplexvolumes}),
 (\ref{eq:w_1}) and (\ref{eq:w_{k+1}}) now imply
$$ w_1 \,\,> \,\, w_{k+1}. $$
This means that $w \not= (\lambda,\lambda,\ldots,\lambda)$ for all $\lambda > 0$.
We conclude that it is impossible to get a nontrivial subdivision
of $X$ as the optimal solution when all the weights are equal.
\end{proof}

We now show that the result of Theorem \ref{thm:d+2points} 
is the best possible in the  following sense.

\begin{theorem}
\label{thm:d+3points}
For any integer $d \geq 2$, there exists a configuration of
$n=d+3$ points in $\RR^d$ for which the
optimal subdivision with respect to unit weights is non-trivial.
\end{theorem}

The hypothesis $d \geq 2$ is essential in this theorem.
Indeed, for $d=1$ it can be shown, 
using the formulas in Example \ref{ex:d=1},
that the flat subdivision is optimal for any
configuration of $d+3=4$ points 
on the line $\mathbb R$ with unit weights.
 Here is an illustration of Theorem~\ref{thm:d+3points}.
 
 \begin{figure}[h]
\begin{center}
\includegraphics[width = 0.35\textwidth]{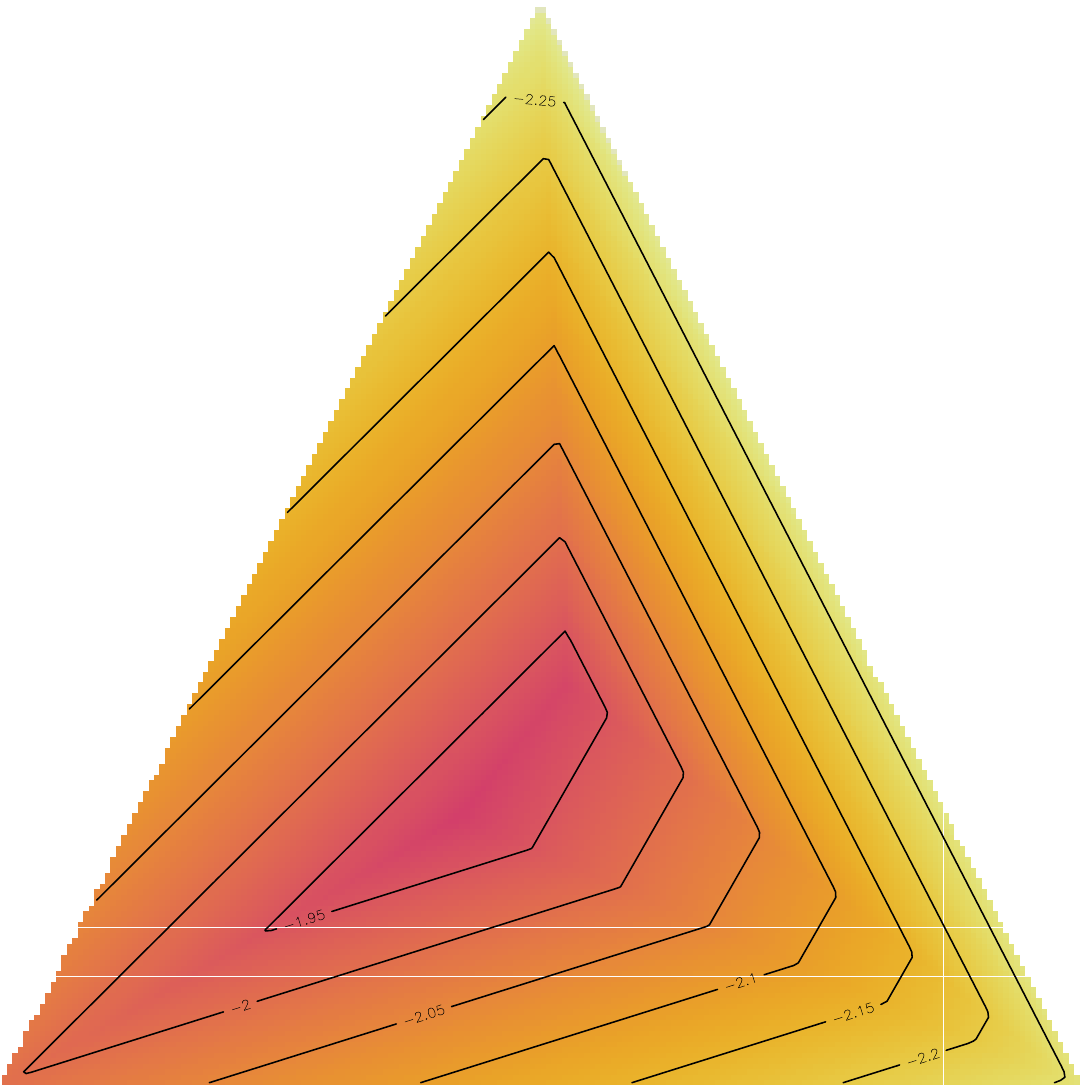} $\quad\qquad$
\includegraphics[width = 0.35\textwidth]{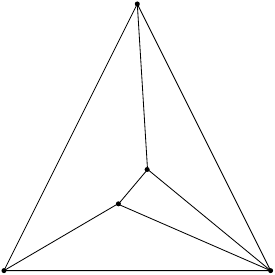}
\end{center}
\vspace{-0.2in}
\caption{The optimal log-concave density for the  five data points in (\ref{eq:fivepoints}) with unit weights.
\label{fig:fivepoints}}
\end{figure}

\begin{example} \rm Fix unit weights on the following five points in the plane:
\begin{align}\label{eq:fivepoints}
 X \,\, = \,\, \bigl(\, (0,0), \,(40, 0),\, (20, 40),\, (17, 10),\, (21, 15) \,\bigr).
\end{align}
Using {\tt LogConcDEAD} \cite{CGS}, we find that 
the optimal subdivision equals $\{124, 245, 235, 1345\}$.
\end{example}

To derive Theorem \ref{thm:d+3points}, we first study
the following configuration of $d+2$ points in $\RR^d$:

\begin{equation}
\label{eq:specialconfig}
X\,\, =\,\, \biggl( e_1\,,\,e_2\,,\,\ldots\,,\,  e_d\,,\, 0\,,\,\,  \frac1{d+1} \sum_{i=1}^de_i \,\biggr). 
\end{equation}

\begin{lemma}
Let $\alpha > 0$ and assign weights as follows to the configuration $X$ in (\ref{eq:specialconfig}):
 \begin{align}\label{eq:equalWeightsFormula}
 w_1=w_2 =\cdots = w_{d+1}  > 0,
 \text{ and } \,\,w_{d+2} \,=\,
 w_1\frac{(d+1)e^{\alpha} H(-\alpha, -\alpha,\dots, -\alpha)}{dH(\alpha, 0,\dots,0)}.
 \end{align}
 Then the optimal heights satisfy $\,y_1= y_2 = \cdots = y_{d+1}\,$ and $\,y_{d+2} = y_1+\alpha$.
 \end{lemma}
 
\begin{proof} 
Let $\mathcal D = \{1,\dots, d+2\}$ and fix $w$ as in (\ref{eq:equalWeightsFormula}).
The volumes $\text{vol}(\mathcal D\backslash\{i\})$ are equal for $i\in\{1,\dots, d+1\}$. Set $\sigma = \text{vol}(\mathcal D\backslash\{i\})$. We will show that  the heights $y_1=\cdots=y_{d+1} = y$ and $y_{d+2} = y+\alpha$ solve the Lagrange multiplier equations \eqref{weightsFormula} for our 
optimization problem, assuming that $\Delta$ is the
 triangulation $\{\mathcal D\backslash \{1\},\ldots, \mathcal D\backslash \{d{+}1 \}\}$. 
Indeed, from \eqref{weightsFormula} we derive
$$ \begin{matrix}
 & w_i &= &
d \cdot \sigma \cdot e^y \cdot H(\alpha, 0, \ldots, 0) \quad & \hbox{
for $i\leq d+1$}  \\
\hbox{and}  \qquad & w_{d+2} 
 & = & (d{+}1)\cdot \sigma \cdot e^{y+\alpha} \cdot H(-\alpha, \dots, -\alpha). & 
 \end{matrix} $$
By taking ratios, we now obtain  \eqref{eq:equalWeightsFormula}.
Of course, the weights must be scaled so that they sum to one.
Since $\alpha > 0$, the subdivision induced by $y$ is indeed 
$\{\mathcal D\backslash 1, \ldots, \mathcal D\backslash \{d{+}1 \} \}$.
\end{proof}

We now note that,  by Lemma~\ref{lem:surprise},
$$e^\alpha \cdot H(-\alpha, \dots, -\alpha) - H(\alpha, 0, \ldots, 0) 
\,\,=\,\, \alpha \frac{\partial H}{\partial \alpha}(\alpha, 0, \ldots, 0).$$
This is positive for $\alpha > 0$, zero for $\alpha = 0$, and negative for $\alpha < 0$. The
first case implies:

\begin{corollary}  \label{cor:setup}
Fix the configuration $X$ in (\ref{eq:specialconfig}) and suppose that
 $w_1=\cdots=w_{d+1}$. Then
   $\frac{w_{d+2}}{w_1} > \frac{d+1}d$ if and only if
the optimal subdivision is the  triangulation $\{\mathcal D\backslash \{1\},
\ldots, \mathcal D \backslash \{ d{+}1 \}\}$.
\end{corollary}

We are now prepared to pass from $d+2$ to $d+3$ points,
and to offer the missing proof.

\begin{proof}[Proof of Theorem \ref{thm:d+3points}]
We use Corollary \ref{cor:setup} with $\frac{w_{d+2}}{w_1} = 2$.
 This is strictly bigger than $\frac{d+1}{d}$ whenever $d\geq 2$. We
 redefine $(X,w)$ by splitting the last point $x_{d+2}$ into two nearby points
 with equal weights. Then $n=d+3$ and the optimal subdivision is non-trivial.
 This holds because, for any fixed $w \in \RR^n$, the set of $X$
 whose optimal subdivision is trivial is described by the
 vanishing of continuous functions. It is hence closed in the space of 
 configurations.
\end{proof}

We conclude this paper with a pair of challenges for Nonparametric Algebraic Statistics.

\begin{problem}  What is the smallest size $n $ of a configuration 
$X$ in $\RR^d$ such that the optimal subdivision of $\,X$ with  unit weights
has at least $c$ cells? \hfill This  $n$ is a function of $c$ and $d$.  \\
We just saw that $n(2,d) = d+3$ for $d \geq 2$.  \hfill
Determine upper and lower bounds for $n(c,d)$.
\end{problem}

We can also ask for a characterization of combinatorial types
of triangulations that are realizable as in Figures \ref{fig:octahedron}
and \ref{fig:fivepoints}.
Such a triangulation in $\RR^d$ is obtained by removing
a facet from a $(d{+}1)$-dimensional simplicial polytope with $\leq n$ vertices.
If we are allowed to vary $w \in \RR^n$, then Theorem \ref{thm:converse} 
tells us that all simplicial polytopes have such a realization.
Hence, in the following question, we seek
configurations $X$ in $\RR^d$ with  $w = \frac{1}{n}(1,\ldots,1)$.

\begin{problem}
Which simplicial polytopes can be realized by
points in $\RR^d$ with unit weights?
\end{problem}

For example, the octahedron can be realized with unit weights, as 
was seen in Figure \ref{fig:octahedron}.

\bigskip \bigskip 

\noindent
{\bf Acknowledgements.} We thank Donald Richards for very helpful discussions regarding Proposition~\ref{prop:magic}.
Bernd Sturmfels was
partially supported by
the Einstein Foundation Berlin and the NSF (DMS-1419018). Caroline Uhler was partially supported by DARPA (W911NF-16-1-0551), NSF (DMS-1651995) and ONR (N00014-17-1-2147).

\bigskip

\begin{small}

\end{small}

\bigskip \bigskip

\noindent
\footnotesize {\bf Authors' addresses:}

\smallskip

\noindent Elina Robeva,
Massachusetts Institute of Technology,
Department of Mathematics, {\tt erobeva@mit.edu}

\noindent Bernd Sturmfels,
 \  MPI-MiS Leipzig, {\tt bernd@mis.mpg.de} \ and \
UC  Berkeley,  {\tt bernd@berkeley.edu}

\noindent Caroline Uhler,
Massachusetts Institute of Technology,
            IDSS and EECS Department,
            {\tt cuhler@mit.edu}.

\end{document}